\documentclass[11pt]{article}
\usepackage[english]{babel}
\usepackage[a4paper,margin=2.54cm]{geometry}

\usepackage{authblk}

\usepackage{amsmath}
\usepackage{amsthm}
\usepackage{amsfonts}
\usepackage{graphicx}
\usepackage[dvipsnames]{xcolor}
\usepackage{tikz}
\usepackage{subcaption}
\usepackage{colortbl}
\usepackage{times}

\usetikzlibrary{arrows.meta, backgrounds}
\usetikzlibrary{positioning}

\usepackage{algorithm}
\usepackage[noend]{algpseudocode}
\usepackage{subcaption}
\usepackage{xcolor,colortbl}
\usepackage{xspace}
\usepackage[round,authoryear]{natbib}

\newcommand{\cellshade}{\cellcolor{lightgray}}
\newcommand{\efx}{EFFX\xspace}
\newcommand{\efOne}{EFF1\xspace}

\usepackage{tikz}

\DeclareMathOperator{\aaa}{\mathcal{A}}
\DeclareMathOperator{\bbb}{\mathcal{B}}
\DeclareMathOperator{\pplus}{\cup}
\DeclareMathOperator{\mminus}{\setminus}

\title{\textsc{Fairness under Equal-Sized Bundles:}\\ Impossibility Results and Approximation Guarantees}

\author[1]{Alviona Mancho}
\author[1,2,3]{Evangelos Markakis}
\author[2]{Nicos Protopapas}

\affil[1]{Athens University of Economics and Business, Greece}
\affil[2]{Archimedes, Athena Research Center, Greece}
\affil[3]{Input Output Global (IOG), Greece}

\definecolor{customgray}{HTML}{d9d9d9}

\theoremstyle{plain}
\newtheorem{theorem}{Theorem}[section]
\newtheorem{lemma}[theorem]{Lemma}
\newtheorem{corollary}[theorem]{Corollary}
\newtheorem{definition}[theorem]{Definition}

\theoremstyle{remark}

\newtheorem{remark}[theorem]{Remark}

\newtheoremstyle{nonitalic}{3pt}{3pt}{}{0pt}{\bfseries}{.}{ }{}
\theoremstyle{nonitalic}
\newtheorem{example}{Example}

\usepackage{enumitem}
\usepackage{algorithm}
\usepackage{algpseudocode}
\algnewcommand\algorithmicforeach{\textbf{for each}}
\algdef{S}[FOR]{ForEach}[1]{\algorithmicforeach\ #1\ \algorithmicdo}

\usepackage{comment}

\definecolor{AegeanBlue}{RGB}{1, 59, 150} 
\usepackage[
  colorlinks=true,
  citecolor=AegeanBlue,
  linkcolor=teal,
  urlcolor=teal
]{hyperref}

\begin{document}
\date{}
\maketitle

\begin{abstract}
    We study the fair allocation of indivisible goods under cardinality constraints, where each agent must receive a bundle of fixed 
    size. This models practical scenarios--such as assigning shifts or forming equally sized teams. 
    Recently, variants of envy-freeness up to one/any item (EF1, EFX) were introduced for this setting, based on flips or exchanges of items. Namely, one can define envy-freeness up to one/any flip (\efOne, \efx), meaning that an agent $i$ does not envy another agent $j$ after performing one or any one-item flip between their bundles that improves the value of $i$.
    
    We explore algorithmic aspects of this notion, and our contribution is twofold: we present both algorithmic and impossibility results, highlighting a stark contrast between the classic EFX concept and its flip-based analogue. First, we explore standard techniques used in the literature and show that they fail to guarantee \efx approximations. On the positive side, we show that we can achieve a constant factor approximation guarantee when agents share a common ranking over item values, based on the well-known envy cycle elimination technique. This idea also leads to a generalized algorithm with approximation guarantees when agents agree on the top $n$ items and their valuation functions are bounded. Finally, we show that an algorithm that maximizes the Nash welfare guarantees a 1/2-\efOne allocation, and that this bound is tight.
\end{abstract}

\thispagestyle{empty}
\newpage

\setcounter{page}{1}
\section{Introduction}

Our work falls under the agenda of fair division with indivisible resources. Fair division has grown considerably in the last years in terms of both theoretical foundations, as can be seen by recent surveys such as \cite{DBLP:journals/ai/AmanatidisABFLMVW23}, but also in terms of motivating applications, including course allocation algorithms \cite{budish2011combinatorial}, food donation programs \cite{Mertzanidis0V24}, and many others.
In the context of indivisible resources, the driving force has been the well known by now fact that more traditional fairness notions, such as envy-freeness and proportionality fail to exist. This gradually led to a quest for defining new solution concepts, that are more appropriate for allocating indivisible items to a set of agents. 

Among the plethora of criteria that have been studied in the literature, our work is mostly related to the prominent notions of EF1 (\emph{envy-freeness up to one good}), defined by~\cite{budish2011combinatorial}, and EFX (\emph{envy-freeness up to any good}) defined in~\cite{gourves2014near,caragiannis2019unreasonable}. 
Both notions are defined with respect to a thought experiment for the agents. More precisely, EF1 demands that if an agent $i$ envies another agent $j$, then she stops being envious of $j$, after removing one item from the bundle of agent $j$. EFX is a stronger notion and demands that an agent $i$ stops being envious of another agent $j$, after removing any item from the bundle of agent $j$. 
An EF1 allocation always exists and can be computed efficiently \cite{lipton2004approximately}, but on the other hand the existence of EFX allocations is guaranteed only for some special cases. In fact it forms one of the greatest open problems in the field of fair division. As a result, and as a way to alleviate the absence of exact EFX allocations till now, there has been a steady stream of works that have focused on algorithms for deriving approximately EFX assignments. 

The above thought experiments for EF1 and EFX are meaningful for defining such relaxations of envy-freeness in settings without any further constraints on the allowed allocations. There are scenarios however where different relaxations of envy-freeness could be more appropriate. In particular, the focus of our work is on instances where all agents have to receive a bundle of the same size. This is naturally applicable for example in settings where a set of employees need to be assigned shifts or rotations (e.g., in a hospital), or in cases where one has to pick teams of equal size (with the items here being the candidate team members). 

The first study of this constrained model within fair division was by \citet{FGM14}, which focused on the concept of maxmin fairness. To our knowledge, the recent work of~\citet{Bog2024} is the first to systematically consider a variety of fairness criteria tailored to this setup. In particular, \cite{Bog2024} adapted the definitions of EF1 and EFX to {\it envy-freeness up to one (resp. any) flip}. Essentially, the thought experiment now is that an agent $i$ is happy if the envy towards an agent $j$ is eliminated after performing an exchange of one item from each other's bundles. This seems more suitable for such constrained problems, because simply removing an item from someone's bundle does not yield a feasible bundle. However, the existence of allocations under this new definition of EFX remains open. Even further, no approximation 
algorithms are yet known for this new concept. 

\subsection{Contribution}

In this work, we explore algorithmic aspects of the recently introduced notions of envy-freeness up to one (resp. any) flip, which we denote by \efOne and \efx respectively, to distinguish them from the standard EF1 and EFX criteria. We are particularly interested in the existence of efficient algorithms for exact and approximate \efx allocations for additive valuation functions. 

We start with some warm-up results in Sections~\ref{sec:comparison} and~\ref{sec:RR}. We first compare the notion of \efx with the EFX criterion and demonstrate that they are generally incomparable. 
We continue then in Section~\ref{sec:RR}, where we focus on whether the Round-Robin algorithm and some of its generalizations, i.e., using different picking sequences, can provide any guarantees. One can easily see that an \efOne allocation can be computed efficiently by the Round-Robin algorithm. Moving to the \efx notion, we prove that such allocations can be computed by such algorithms, when the bundles are of size two each. 
For bundles of higher size however, we show a severe negative result, that any algorithm within the class of generalized Round-Robin algorithms (where the order of agents in each round may differ) cannot guarantee any approximation for \efx. This holds even for identical valuations. 

In Section~\ref{sec:ece}, we then study variations of the Envy Cycle Elimination (ECE) algorithm, a dominant tool used in algorithms with indivisible items. In contrast to the standard EFX concept where this algorithm computes a $1/2$-EFX approximation, here we show that its natural adaptation to our constrained setting cannot guarantee any approximation at all for general additive valuations. 
On the positive side, we derive approximation guarantees for three classes of instances, as follows.
\begin{itemize}[noitemsep,leftmargin=*]
    \item {\bf Ordered valuations.} First, we show that the adaptation of the ECE algorithm in our setting achieves  a $1/2$-\efx guarantee for instances where all the agents agree on the ranking of the items, from the most valuable to the least valuable one. This is a commonly studied special case of the problem and our result comes in contrast to the unconstrained setting, where the ECE algorithm produces an exact  EFX allocation, as shown by~\cite{plaut2020almost}.
    \item {\bf Agreement on the top $n$ items.} We then relax the common ranking assumption and focus on the case where all agents only agree on what is the set of the top $n$ most valuable items, where $n$ is the number of agents. For the unconstrained model, a $2/3$-EFX guarantee is known by \cite{markakis23improved}. Our constrained model makes this case more challenging as well. We significantly modify the ECE algorithm by allowing certain agents to give away previously acquired items and by also allowing certain envied agents to also receive new items. Our main result is that this new algorithm achieves a $\min\{1/3,1/(\rho+1)\}$-EFFX allocation, where $\rho$ is the ratio among the maximum and minimum value within the top $n$ items, over all agents. 
    \item {\bf Bounded ratio within the top $n$ items.} Our last guarantee concerns instances where we only have a bound of $\rho$ for the maximum ratio between any two items among the $n$ most valuable items of each agent. We show that the same algorithm as before obtains a $1/(\rho+2)$-approximation.
\end{itemize} 
 
Finally, in Section~\ref{sec:PO}, we consider the combination of attaining some form of efficiency together with \efOne or \efx. We examine three predominant efficient methods, namely, (i) computing a social welfare maximizing allocation, (ii) a leximin order based allocation and (iii) computing a Nash welfare maximizing allocation. 
We expand a result from~\cite{Bog2024} and show that the first two methods cannot guarantee better than $O(1/k)$-\efOne allocations, while the third one cannot guarantee better than $1/2$-\efOne allocations. Our main positive result is that the latter is tight, i.e., any Nash welfare optimal solution is also $1/2$-\efOne. 
This again reveals a difference with the unconstrained setting where a Nash welfare optimal solution is also EF1 \cite{caragiannis2019unreasonable}. 
Furthermore, we also show that the leximin method cannot yield an \efx allocation for instances with ordered valuations, filling a gap left by~\cite{Bog2024}. 

Overall, our results demonstrate that despite the similarity in the definitions, there exist major differences between \efx and the standard notion of EFX. The constraint of equal cardinality bundles introduces technical challenges that make it more intriguing (and currently elusive) to have approximation algorithms for general additive valuations. Nevertheless, one can still 
obtain guarantees for some well-studied families of valuations as outlined above.
 
\subsection{Further related work}

The most closely related work to ours is \cite{Bog2024}. Unlike the more standard models, that do not impose constraints on the size of individual bundles, their work explored fairness when each agent is allocated a bundle of size exactly $k$. They adapted many of the popular fairness notions to incorporate the flip of items, such as proportionality and envy-freeness up to one or up to any flip. Notably, they proved that an envy-free up to any flip allocation always exists under one of the following conditions: (i) the agents have identical valuations, (ii) the agents have binary utilities, or (iii) there are two agents.  However, their work did not study the approximability of these fairness concepts.

In the unconstrained setting, there has been a surge of works on existence and algorithmic results for fairness with indivisible items, especially regarding the EFX notion, which was introduced in \citet{gourves2014near,caragiannis2019unreasonable}.
The first set of positive results were provided by \citeauthor{plaut2020almost}, proving that EFX allocations exist when all agents have identical, not necessarily additive, valuation functions. They also established existence and efficient computation, when all agents agree on the ranking of the items w.r.t. their value. In \cite{ChaudhuryGM24}, existence was established for three agents with additive valuations. A simpler proof, which also allows for some further generalizations, was more recently obtained by \citet{AkramiACGMM25}.
Existence has also been established for instances with three distinct values for the goods \cite{AFS24}, improving on the previous result for bivalued instances \cite{amanatidis2021maximum}.
In light of the challenges of satisfying exact EFX, there is also a stream of works on approximation algorithms, starting with \citeauthor{plaut2020almost}, who showed a $1/2$-EFX approximation (albeit in exponential time) for subadditive valuations. This was later improved to a polynomial-time algorithm in \cite{10.5555/3367032.3367053}. Currently, the best-known approximation is $\phi - 1 \approx 0.618$, due to \cite{amanatidis2020multiple}. Moreover, there have been already a few improved approximation guarantees for several special cases. In \cite{markakis23improved} a $2/3$-EFX algorithm is proposed for a scenario in which agents agree on the top $n$ items, where $n$ is the number of agents.
Further special cases that attain a $2/3$-approximation are also established in \cite{AFS24}.
In addition to these results, a general framework for constructing approximation algorithms for EFX is discussed in \cite{markakis23improved} and \cite{farhadi2021almost}. 

Similar questions arise in scenarios where items are viewed as chores (negatively valued by the agents). Chores have been proven to behave differently than goods. In \cite{DBLP:journals/corr/abs-2406-10752}, the question of the existence of EFX allocations was answered negatively for chores when agents have superadditive valuation functions. The state of the art result for EFX is a 4-EFX guarantee, as presented in \cite{garg2024fairdivisionindivisiblechores}, where the first constant-factor approximation for EFX was introduced, while their work also extended to other fairness notions. The existence of EFX allocations for chores has been established in special cases, notably when agents have additive valuations and the number of chores does not exceed twice the number of agents, as demonstrated in \cite{kobayashi2023efxallocationsindivisiblechores}. For an extensive discussion on the fair division of indivisible chores, we refer the reader to \cite{math11163616}. In scenarios where items may be perceived as both goods and chores--commonly referred to as \emph{mixed manna}--it has been shown that an \textsc{EFX} allocation does not necessarily exist. This impossibility has been established for two agents with identical, non-additive, non-monotone valuation functions in~\cite{Berczi20} and for agents with additive (specifically, lexicographic) preferences in~\cite{hosseini2022fairlydividingmixturesgoods}.

Moreover, the fairness notion of \emph{envy-freeness up to transferring any good or chore} (tEFX) has been introduced in~\cite{barman2023parameterizedguaranteesenvyfreeallocations} and~\cite{yin2022envyfreeallocationchores}, respectively. This concept is similar to \efx in that the underlying thought experiment involves removing an item from one agent’s bundle and transferring it to another. However, it differs from our work, and from the notion of rational flips that we use, in the sense that tEFX considers only the transfer of a single item to an agent’s bundle, rather than the exchange of item pairs. In the context of chores, \cite{afshinmehr2024approximateefxexacttefx} shows that a tEFX allocation is attainable for three agents when one agent has an additive valuation function that is $2$-ratio bounded, and the remaining agents have general monotone valuation functions.

Apart from EF1 and EFX, there is a broader range of fairness concepts that have been considered, such as Maximin Share (MMS), Pairwise Maximin Share (PMMS), and Groupwise Maximin Share (GMMS) fairness, introduced in \cite{budish2011combinatorial}, \cite{caragiannis2019unreasonable} and \cite{barman2018groupwise}, respectively. Additionally, relaxations of proportionality have been proposed, such as proportionality up to one (Prop1) and up to any good (PropX), defined in \cite{conitzer2017fair} and \cite{DBLP:journals/orl/AzizMS20}, respectively. For a detailed discussion of further significant fairness notions and unresolved questions, we refer the reader to the survey in \cite{DBLP:journals/ai/AmanatidisABFLMVW23}.

Finally, regarding the use of constraints in fair division, \citet{FGM14} also considered the same model as ours but with the different objective of maximin fairness, for which they obtain approximation and exact algorithms.
Other models of cardinality constraints have also been considered in fair division. For example \cite{BB18} study the case where the available items are grouped into categories and the bundle of each agent should respect a bound on items from each category. This setting is incomparable to ours. For an overview of further problems involving constrained fair division we refer to \cite{Suksompong21}.

\section{Preliminaries}\label{sec:prel}
For any $z \in \mathbb{N}_{> 0}$ we use $[z]$ to denote the set $\{1,2,\dots,z\}$.
We consider a set of agents $N = [n]$ and a set $M=[kn]$ of $kn$ indivisible items for some $k,n \in \mathbb{N}_{\geq 2}$.
An allocation $\aaa$ in our model is any ordered partition of the items into $n$ subsets, $\aaa = (A_1, \dots, A_n)$, where $A_i$ is the bundle of agent $i$ and such that each agent must receive exactly $k$ items, i.e., $|A_i|=k$. 

We consider agents with cardinal valuation functions. 
In particular, we assume that valuation functions are non-negative\footnote{In this paper we only care about goods.} and additive, i.e., each agent $i \in [n]$ associates a value $v_i(\{g\})\geq 0$ for each item $g \in [kn]$, and for a given bundle $A$, $v_i(A)=\sum_{g \in A} v_i(\{g\})$. From this point and further, we will use the abbreviation $g$ to denote a singleton set $\{g\}$, so that e.g., $v_i(g)=v_i(\{g\})$ or $v_i(A\cup g \setminus g') = v_i(A \cup \{g\} \setminus \{g'\})$, for the sake of simplicity.

An ideal solution concerning fairness is that no agent prefers another agent's bundle to their own. Formally,

\begin{definition}[Envy freeness-EF]
An allocation $\aaa$ is envy-free (EF) if for every pair of agents $i,j$, it holds that $v_i(A_i) \ge v_i(A_j)$.
\end{definition}

It is well-known that envy-free allocations do not always exist. 
Therefore several relaxations have been considered as alternative solutions. Among these, the two most related to our work are the well known criteria of EF1 and EFX, defined as follows.

\begin{definition}\label{def:ef1}
An allocation $\aaa$ is 
\begin{itemize}
    \item envy-free up to one good (EF1) if for every pair of agents $i,j$, there exists a good $g\in A_j$, such that $v_i(A_i) \ge v_i(A_j \setminus g)$.
    \item envy-free up to any good (EFX) if for every pair of agents $i,j$, and for every $g\in A_j$, it holds that $v_i(A_i) \ge v_i(A_j \setminus g)$.
\end{itemize}
\end{definition}

The intuition behind the EF1 notion, which was defined by \citet{budish2011combinatorial}, is that the agents cannot be too envious, in the sense that there always exists a single item whose removal can eliminate envy from one agent to another. 
The EFX notion (defined in \citet{caragiannis2019unreasonable} and also in \citet{gourves2014near}) is stronger since the difference w.r.t. EF1 is the switch of the quantifiers, so that envy can be eliminated by the removal of any single item.

The above relaxations follow the thought experiment that one can discard one item from another agent's bundle. 
In our scenario however, where every agent has to receive exactly $k$ items in her bundle, this may not be suitable when an agent wants to compare her assignment against other agents. 
As a result, \citet{Bog2024} proposed similar relaxations but using the idea of a {\it flip} instead of item removals. 

To become more precise, given agents $i,j$ and bundles $A_i$ and $A_j$ from the same allocation $\aaa$, we say that the pair $(a,b)$, such that $a \in A_i$, $b \in A_j$ is a \emph{rational flip} w.r.t. $i$ if $v_i(b) > v_i(a)$. Intuitively, this notion implies that it makes sense for $i$ to exchange $a$ with $b$. It is not hard to see that under additive valuations, when an agent envies another, there exists at least one rational flip between them. The natural adjustment of the EF1 criterion to incorporate flips instead of a single item removal, results to a new fairness notion of \emph{envy-freeness up to one flip}, defined as follows.

\begin{definition}[\efOne]
An allocation $\aaa$ is envy-free up to one flip (\efOne) if for every pair of agents $i,j$, either $v_i(A_i) \geq v_i(A_j)$, or there exists a rational flip $a \in A_i$, $b \in A_j$, such that $v_i(A_i \pplus b \mminus a) \geq v_i(A_j \pplus a \mminus b)$.
\end{definition}

The intuition behind \efOne allocations is that there always exists a pair of items whose exchange can eliminate envy from one agent to another. This is an efficiently computable fairness notion, as we will demonstrate later on.

Towards coming closer to envy-freeness, the analogue of EFX with the use of flips can also be similarly defined, and this is the main notion of interest for our work. The difference with \efOne is that we require that envy can be eliminated by {\it any} rational flip.

\begin{definition}[\efx]
An allocation $\aaa$ is envy-free up to any flip (\efx) if for every pair of agents $i,j$, either $v_i(A_i) \geq v_i(A_j)$ or for every pair of items $(a, b)$, with $a \in A_i$, $b\in A_j$, that forms a rational flip w.r.t. agent $i$, it holds that $v_i(A_i \pplus b \mminus a) \ge v_i(A_j \pplus a\mminus b)$.
\end{definition}

It is still unknown if \efx allocations exist, beyond some special cases. For this reason, we will also study approximate versions of \efx. Although there are multiple ways of defining an approximation notion, we will stick to the multiplicative version, similar to the approximation notion for EFX, defined by~\cite{plaut2020almost} and used in most previous works as well: 

\begin{definition}[$\gamma$-\efx]
An allocation $\aaa$ is $\gamma$-\efx for $\gamma \in [0,1]$ if for every pair of agents $i,j$ either $v_i(A_i) \geq v_i(A_j)$, or for every pair of items $(a, b)$, with $a \in A_i$, $b\in A_j$, that forms a rational flip w.r.t. agent $i$, it holds that $v_i(A_i \pplus b \mminus a) \ge \gamma v_i(A_j \pplus a\mminus b)$.
\end{definition}

Hence, our goal is to obtain $\gamma$-\efx allocations, with $\gamma$ as close to 1 as possible. In the same manner, one can also define approximate versions for other concepts (e.g. $\gamma$-EF, $\gamma$-EFX or $\gamma$-\efOne).

\subsection{Warming-up: Observations and comparisons}\label{sec:comparison}
Before presenting our main results, we begin with a series of basic yet insightful observations. The following theorem establishes a trivial yet informative benchmark, highlighting the limits of what we can hope to achieve. The second part, showing that for $k = 2$ the \efOne property holds trivially, is already known from~\cite{Bog2024}; a short proof is included for the sake of completeness.

\begin{theorem}[partially known from \cite{Bog2024}]
    Any allocation is either EF or $1/k$-EFF1. Furthermore, when $k=2$, any allocation is \efOne.
\end{theorem}

\begin{proof}
    Consider an allocation $\aaa$. If $\aaa$ is EF, then the claim holds trivially. Otherwise, there exists a pair of agents $i$ and $j$ such that agent $i$ envies agent $j$, i.e., $v_i(A_i) < v_i(A_j)$. Let $g_{\text{max}}$ be the item in $A_j$ that agent $i$ values the most. By an averaging argument, we have:
    \begin{align*}
        v_i(g_{\text{max}}) &\geq \frac{1}{k} \cdot v_i(A_j).
    \end{align*}

    Since agent $i$ envies agent $j$, there must exist an item $g \in A_i$ such that $v_i(g) < v_i(g_{\text{max}})$. Performing the rational flip $(g, g_{max})$, yields:
    \begin{align*}
        v_i(A_i \cup g_{\text{max}} \setminus g) &\geq \frac{1}{k} \cdot v_i(A_j) \geq \frac{1}{k} v_i(A_j \setminus g_{\text{max}} \cup g).
    \end{align*}

  \noindent  The latter implies that the allocation satisfies $1/k$-EFF1.

    For the special case where $k = 2$, consider the bundles $A_i = \{a, b\}$ and $A_j = \{c, d\}$. Assume that $v_i(a) \geq v_i(b)$ and $v_i(c) \geq v_i(d)$. Then, under the assumption that $i$ envies $j$, performing the rational flip $(b,c)$ yields new bundles $A_i' = \{a,c\}$ and $A_j' = \{b,d\}$. We have 
    \begin{align*}
        v_i(A_i') = v_i(a) + v_i(c) \geq  v_i(b) + v_i(d) = v_i(A_j').
    \end{align*}
    
    Therefore, the allocation satisfies \efOne in this case.
\end{proof}

Next, we present simple examples to demonstrate that the fairness notions \efOne and \efx, as studied in this work, are incomparable with the standard EF1 and EFX criteria, respectively. As shown in \cite{Bog2024}, an EF1 allocation is not necessarily \efOne, demonstrated through an instance with $n = 2$, $k = 3$ and identical valuations. They also construct a similar instance with $n = 2$ and $k = 3$ to show that an \efOne allocation is not necessarily EF1. To reinforce this point, we provide an even smaller example with $n = 2$ and $k = 2$ that illustrates the latter direction. In addition, we examine the relationship between EFX and \efx using analogous examples. In the valuation tables that follow, $C > 1$ denotes a large constant, and $0 < \epsilon \ll 1$ represents a small positive value.

\begin{table}[tbp]
\centering

\begin{subtable}{0.3\textwidth}
\centering
\begin{tabular}{l| c c c c}
     & $g_1$ & $g_2$ & $g_3$ & $g_4$ \\
    \hline
    $v(\cdot)$ & \cellshade$2C$ & \cellshade$1+2\epsilon$ & 1 & $\epsilon$ \\
    \phantom{$v_1(\cdot)$} & \phantom{$x$} & \phantom{$x$}  & \phantom{$x$}  & \phantom{$x$} 
\end{tabular}
\caption{Example~\ref{exm:eff1-to-ef1}}\label{tb:example1}
\end{subtable}
\hspace{0.5cm}
\begin{subtable}{0.3\textwidth}
\centering
\begin{tabular}{l| c c c c}
     & $g_1$ & $g_2$ & $g_3$ & $g_4$ \\
    \hline
    $v_1(\cdot)$ & \cellshade10 & 6 & 4 & \cellshade1 \\
    $v_2(\cdot)$ & $10+\epsilon$ & \cellshade10 & \cellshade1 & 2
\end{tabular}
\caption{Example~\ref{exm:efx-to-efxf}}\label{tb:example2}
\end{subtable}
\hspace{0.5cm}
\begin{subtable}{0.3\textwidth}
\centering
\begin{tabular}{l| c c c c}
     & $g_1$ & $g_2$ & $g_3$ & $g_4$ \\
    \hline
    $v(\cdot)$ & \cellshade$2C$ & $1+\epsilon$ & 1 & \cellshade$\epsilon$ \\
    \phantom{$v_1(\cdot)$} & \phantom{$x$}  &\phantom{$x$}  & \phantom{$x$}  & \phantom{$x$} 
\end{tabular}
\caption{Example~\ref{exm:efxf-to-efx}}\label{tb:example3}
\end{subtable}

\caption{The valuations in Examples~\ref{exm:eff1-to-ef1}, \ref{exm:efx-to-efxf}, and \ref{exm:efxf-to-efx}.}
\label{tab:three}
\end{table}

\begin{example}[An \efOne allocation is not necessarily EF1]\label{exm:eff1-to-ef1}
    Assume there are $4$ items with identical values for both agents. Let $v(\cdot)$ be the common valuation function. The valuation is shown in Table~\ref{tb:example1}.

    Consider the allocation $\aaa = (A_{1},A_{2})$, where $A_1 = \{g_1, g_2\}$ and $A_2 = \{g_3, g_4\}$, shown shaded. This is an \efOne allocation, since agent~$1$ does not envy agent~$2$ ($v(A_1) = 2C+ 1 + 2\epsilon > v(A_2) = 1+\epsilon$), whereas agent~$2$ envies agent~$1$, but any rational flip that involves item $g_1$ eliminates the envy. However, this allocation is not EF1: No item $g \in A_1$ can be removed to eliminate envy, since
    \begin{align*}
        v(A_2) = 1 + \epsilon < v(A_1 \mminus g_1) = 1 + 2\epsilon \quad \text{and}\quad
        v(A_2) 
        = 1 + \epsilon < v(A_1 \mminus g_2) = 2C.
    \end{align*}
\end{example}

\begin{example}[An EFX allocation is not necessarily \efx]\label{exm:efx-to-efxf}
 Table~\ref{tb:example2} shows the valuations over $4$ items for agents $1$ and $2$.  
 
    Consider the allocation $\aaa=(A_{1},A_{2})$, where $A_1=\{g_1, g_4\}$ and $A_2=\{g_2, g_3\}$, shown shaded. This is an EFX allocation, since agent $1$ does not envy agent $2$ ($v_1(A_1) = 11 > v_1(A_2) = 10$), whereas agent $2$ envies agent $1$ ($v_2(A_2) = 11 < v_2(A_1) = 12+\epsilon$), but removing either $g_1$ or $g_4$ from $A_1$ eliminates the envy. However, this allocation is not \efx. There exists a rational flip, namely ($g_2, g_1$), where $g_1 \in A_1$, $g_2 \in A_2$ and $v_2(g_1) > v_2(g_2)$ violating \efx, since
    \begin{align*}
        v_2(A_2 \pplus g_1 \mminus g_2) = 11 + \epsilon < v_2(A_1 \mminus g_1 \pplus g_2) = 12.
    \end{align*}
\end{example}

\begin{example}[An \efx allocation is not necessarily EFX]\label{exm:efxf-to-efx}
  Assume there are $4$ items with identical values for both agents, shown in Table~\ref{tb:example3}. Let $v(\cdot)$ be the common valuation function.
  
    Consider the allocation $\aaa=(A_{1},A_{2})$, where $A_1=\{g_1, g_4\}$ and $A_2=\{g_2, g_3\}$, shown shaded. This is an \efx allocation, since agent $1$ does not envy agent $2$ ($v(A_1) = 2C+\epsilon > v(A_2) = 2+\epsilon$), whereas agent $2$ envies agent $1$, but any rational flip must involve item $g_1$, and hence, once performed, it eliminates the envy. However, this allocation is not EFX: There exists an item, namely $g_4 \in A_1$, such that
    \begin{align*}
        v(A_2) = 2+\epsilon < v(A_1 \mminus g_4) = 2C.
    \end{align*}
\end{example}

Observe that in Example~\ref{exm:efxf-to-efx}, there exists only one \efx allocation, specifically the one examined earlier (and its symmetric counterpart). To see why this is the case, suppose w.l.o.g. that agent $1$ receives item $g_1$. Then, agent $2$ will envy agent $1$ regardless of how the remaining items are allocated. The only way for the allocation to satisfy \efx is if all rational flips for agent $2$ involve item $g_1$. This can only occur if agent $2$ is allocated $\{g_2, g_3\}$, which implies that agent $1$ must receive $\{g_1, g_4\}$.
    
\noindent Moreover, observe that
\begin{align*}
    \frac{v(A_2)}{v(A_1)} = \frac{2+\epsilon}{2C+\epsilon} \to \frac{1}{C},
\end{align*}
as $\epsilon$ reaches 0. Hence, this example supports an even stronger conclusion: In this setting, \efx is incompatible with $\gamma$-EFX for any $\gamma \in (0,1]$, meaning that no meaningful approximation of EFX can align with \efx. 

\section{Generalized Round-Robin algorithms}
\label{sec:RR}

We start the technical part of our contribution with analyzing a simple family of algorithms.  

As a warm-up, we show that \efOne allocations can be computed in polynomial time using the well-known Round-Robin algorithm. The Round-Robin rule fixes an ordering of the agents and runs in $k$ rounds; in each round, the agents pick their favorite available item, according to the predefined ordering, which is the same for all rounds. The result is already known by~\cite{Bog2024} and here we provide an alternative proof for the sake of completeness. The proof is deferred to Appendix~\ref{app:RR}.

\begin{theorem}[follows by \cite{Bog2024}]\label{thm:RR:Eff1}
    The allocation computed by the Round-Robin algorithm satisfies \efOne.
\end{theorem}

We now move to the more demanding case of \efx allocations. 
As with the standard EFX notion, it is unknown if such allocations always exist for additive valuations, and this already looks like a challenging open problem. 
They exist however for the special cases of (i) $n=2$, (ii) when all agents have identical valuations, and (iii) for binary valuation functions, as established in~\cite{Bog2024}.

Here we show that \efx allocations always exist when $k=2$, for additive valuations and any number of agents, and can be computed efficiently using a simple adaptation of Round-Robin.

The algorithm fixes an order for the agents and runs in $2$ rounds. In the first round, each agent picks their favorite available item. In the second round agents pick again their favorite available item, but in reverse order. Observe that we arrive at this guarantee only using \emph{ordinal} information about the items -- agents are not required to disclose their exact valuations.

\begin{theorem}\label{theorem: existence-k-2}
    \efx allocations are guaranteed to exist when $k=2$ for additive valuations.
\end{theorem}
\begin{proof}
    Suppose we run one round of the Round-Robin algorithm with a given order of the agents, and then we run one more round with the reverse order. Consider a pair of agents $i, j$, and assume that $i$ envies $j$.
    Since $k=2$ we can write $A_i=\{b_1,b_2\}$ and $A_j=\{a_1,a_2\}$ as the bundles of the two agents produced by the algorithm, where $a_1, b_1$ are allocated in the first round and $a_2, b_2$ in the second round.
    
    We will break the analysis into two cases: whether $i$ picks \emph{before} $j$ in the first round or not.
    \begin{description}[leftmargin=0pt]
        \item[Case 1:] $i$ picks \emph{before} $j$ in the first round. Then $v_i(b_1) \geq v_i(a_\ell)$ for $\ell \in \{1,2\}$, hence $b_1$ cannot be in any rational flip. 
        Hence, we are left with two possible rational flips: $(b_2,a_1)$ or $(b_2,a_2)$. Consider the first of them, $(b_2, a_1)$. If this is indeed a rational flip, then after we perform the flip, agent $i$ has the bundle $\{b_1, a_1\}$, and she likes $b_1$ at least as much as $a_2$ and similarly $a_1$ at least as much as $b_2$. Therefore, she cannot envy $j$ after the flip. 
        The same argument also holds if the flip $(b_2,a_2)$ is rational w.r.t. $i$.
        \item[Case 2:] $i$ picks \emph{after} $j$ in the first round. Hence $v_i(b_1) \geq v_i(b_2) \geq v_i(a_2)$. Clearly, $a_2$ cannot participate in any rational flip. The two possible flips are $(b_1,a_1)$ or $(b_2,a_1)$. Suppose that $(b_1,a_1)$ is a rational flip. Then, after performing the flip, agent $i$ will have the bundle $\{a_1, b_2\}$. Observe now that she likes $a_1$ at least as much as $b_1$ and $b_2$ at least as much as $a_2$, therefore she cannot be envious. The same argument also holds for the other flip. This concludes the proof.
    \end{description}
\end{proof}

\begin{remark}
    Interestingly, Theorem~\ref{theorem: existence-k-2} proves the existence of an \efx allocation for $2n$ items, whereas, for the standard EFX notion, existence is known only for up to $n+3$ items by \cite{mahara2021extension}. 
\end{remark}

We consider now a more general family of algorithms, referred to in \cite{Bog2024} as \textit{generalized Round-Robin}. 
This is a class of algorithms that still work in $k$ rounds and they allocate exactly one item to each player per round. However, it is allowed to use a different order of agents in each round. 
The algorithm in Theorem~\ref{theorem: existence-k-2} falls into this class, but only for $k=2$. One would wonder if we could have further existence guarantees for higher values of $k$. We answer this in the negative, providing below a severe inapproximability result for all these algorithms, even for $k=3$ and $n=2$ agents.

\begin{theorem}\label{theorem: gen-RR-inapprox}
    Any generalized Round-Robin algorithm fails to guarantee $\gamma$-\efx for any $\gamma \in (0,1]$ when $k>2$, even for $n=2$ agents with identical valuations.
\end{theorem}

\begin{proof}
    We construct an instance with 2 identical agents and 6 items. Table~\ref{tab:gen-RR-inapprox} shows the valuations of the 6 items for agents $1$ and $2$, where $C>1$ can be set to be a large constant, and $0 < \epsilon < 1$. Since the agents are identical, let $v(\cdot)$ be their common valuation function.

    \begin{table}[tbp]
        \centering
        \begin{tabular}{l|c c c c c c c c c}
            & $g_1$ & $g_2$ &$g_3$ & $g_4$ & $g_5$ & $g_6$\\
            \hline
             $v(\cdot)$&$3C$ &$1+\epsilon$  &1 &$1-\epsilon$ &$\epsilon$ &0\\
        \end{tabular}
        \caption{The valuations for the proof of Theorem~\ref{theorem: gen-RR-inapprox}}
        \label{tab:gen-RR-inapprox}
    \end{table}
    
    Consider a particular instantiation of generalized Round-Robin, i.e., an order between the two agents in each of the three rounds. Suppose that in the first round, agent 1 picks first. Let $A_1$ and $A_2$ be the bundles produced for the two agents. We will show that none of the possible allocations satisfies \efx.
        
    \textbf{Case 1:} $A_{1} = \{g_1, g_3, g_5\}$ $A_{2} = \{g_2, g_4, g_6\}$. In this case,
    \begin{align*}
        v(A_{1}) = 3C + 1 + \epsilon \text{ and }
        v(A_{2}) = 2.
    \end{align*}
     Hence, agent 1 has no envy towards $2$, whereas 2 envies 1 even after the rational flip $(g_6, g_5)$ w.r.t. $2$. Observe that $v(A_{1} \mminus  g_5 \pplus  g_6) = 3C + 1$ and $v(A_{2} \pplus  g_5 \mminus  g_6) = 2 + \epsilon$.
    
    \textbf{Case 2:} $A_{1} = \{g_1, g_3, g_6\}$ $A_{2} = \{g_2, g_4, g_5\}$. In this case,
    $
        v(A_{1}) = 3C + 1 \text{ and }
        v(A_{2}) = 2 + \epsilon.
    $
     Hence, agent $1$ has no envy towards $2$, whereas $2$ envies $1$, even after the rational flip $(g_4, g_3)$. Observe that $v(A_{1} \mminus  g_3 \pplus  g_4) = 3C + 1 - \epsilon$ and $v(A_{2} \pplus  g_3 \mminus  g_4) = 2 + 2\epsilon$.
    
    \textbf{Case 3:}
    $A_{1} = \{g_1, g_4, g_5\}$ $A_{2} = \{g_2, g_3, g_6\}$. In this case,
    $
        v(A_{1}) = 3C + 1 \text{ and }
        v(A_{2}) = 2 + \epsilon.
    $
     Hence, agent $1$ has no envy towards $2$, whereas $2$ envies $1$ even after the rational flip $(g_6, g_5)$. Observe that $v(A_{1} \mminus  g_5 \pplus  g_6) = 3C + 1 - \epsilon$ and $v(A_{2} \pplus  g_5 \mminus  g_6) = 2 + 2\epsilon$. 
    
    \textbf{Case 4:} $A_{1} = \{g_1, g_4, g_6\}$ $A_{2} = \{g_2, g_3, g_5\}$. In this case,
    $
        v(A_{1}) = 3C + 1 - \epsilon \text{ and }
        v(A_{2}) = 2 + 2\epsilon.
    $
     Again, agent $1$ has no envy towards $2$, but $2$ envies $1$ even after the rational flip $(g_5, g_4)$. Observe that $v(A_{1} \mminus  g_4 \pplus  g_5) = 3C + \epsilon$ and $v(A_{2} \pplus  g_4 \mminus  g_5) = 3$. 
  
    As demonstrated in the above cases, there always exists at least one rational flip which does not eliminate the envy. Therefore, none of the above allocations is \efx. Note that the instance does admit an \efx allocation, namely $A_1 = \{g_1, g_5, g_6\}$ and $A_2 = \{g_2, g_3, g_4\}$, but this cannot correspond to any allocation derived by the generalized Round-Robin algorithms. 

    Furthermore, let $A_{1}'$ and $A_{2}'$ be the bundles of agents $1$ and $2$ respectively, after performing the rational flip suggested in each of the above cases. Since in each case $v(A_{2}') < 3$ and $v(A_{1}') > 3C$, then $v(A_2')/v(A_1') < 1/C$.
    Symmetrically, due to the identical valuations, the same arguments hold if we had assumed that agent 1 picked second in the first round, in which case the role of agents $1$  and $2$ would be switched in the above analysis. Therefore, we have established that generalized Round-Robin algorithms cannot guarantee better than $1/C$-\efx allocations. Since $C$ is a parameter that we can make arbitrarily large, this means that we cannot have any approximation guarantee, for any $\gamma\in (0, 1]$.

    Finally, it is straightforward to generalize this result to instances with any $k>3$ by adding dummy items of value zero.
\end{proof}

\noindent {\bf Using generalized Round-Robin as an initial step in approximation algorithms.} Embedded within the proof of Theorem~\ref{theorem: gen-RR-inapprox} is an even stronger negative indication for the use of Round-Robin to obtain worst-case approximation guarantees for \efx. In unconstrained fair division, several algorithms for approximately EFX allocations are based on using an initial phase where a subset of items (usually at most $2n$ in total) are allocated via generalized Round-Robin, so as to start with a partial $\gamma$-EFX allocation for some $\gamma$. This can be seen in the approximation framework discussed in \cite{markakis23improved} and \cite{farhadi2021almost}.
However, when it comes to \efx, consider the instance of Theorem~\ref{theorem: gen-RR-inapprox} with 6 items, and assume that we assign the four most valuable goods to agents $1$ and $2$ to have a partial \efx allocation. Due to Theorem~\ref{theorem: existence-k-2}, such a partial allocation can be found, namely $A_1 = \{ g_1, g_4\}$ and $A_2 = \{g_2, g_3\}$. 
However, regardless of how we allocate the remaining two items, the final allocation is doomed to lack any approximation guarantee. If we allocate items $g_5$ and $g_6$ to agents $1$ and $2$ respectively, then the full allocation becomes $A_1 = \{g_1, g_4, g_5\}$ and $A_2 = \{g_2, g_3, g_6\}$, which coincides with Case 3 in the proof of Theorem~\ref{theorem: gen-RR-inapprox}. Alternatively, in the opposite case, the full allocation becomes $A_1 = \{g_1, g_4, g_6\}$ and $A_2 = \{g_2, g_3, g_5\}$, which coincides with Case 4 of the same proof. 
In either scenario, no approximation ratio for \efx is guaranteed.
\section{Algorithms based on the envy graph}
\label{sec:ece}

We now move to more powerful algorithms, based on a graph-theoretic representation of allocations.
Given a (possibly partial) allocation $(A_1,...,A_n)$ of the items to the $n$ agents, the \emph{envy graph} is defined as the graph $G=(N,E)$ where a directed edge $(i,j)$ exists in $E$ iff $v_i(A_i) < v_i(A_j)$. 

The ECE procedure, introduced in \cite{lipton2004approximately} is based on the notion of the envy graph and can compute an EF1 allocation in the standard model in polynomial time, even for more general valuation classes than additive. 
We describe briefly the algorithm first in the unconstrained model. The algorithm starts from an empty allocation and runs in rounds. In each round, an unenvied node (i.e., a source node in $G$ associated with the partial allocation at that round) is allocated her most valuable item among the ones that are still available. Observe that we can always guarantee the existence of an unenvied node in each round: whenever no such agent exists, the graph $G$ contains a directed cycle in the form $i_1 \rightarrow i_2 \rightarrow \ldots \rightarrow i_1$. By allocating the bundles backwards along this cycle, the cycle can be eliminated, and by doing the same process for as many times as needed, all the cycles of the graph are eliminated, and an unenvied agent is guaranteed to exist.  For the sake of completeness, the formal description of the ECE algorithm in the unconstrained model can be seen in Appendix \ref{app:ece}.

\begin{theorem}[from \cite{10.5555/3367032.3367053,markakis23improved}]
\label{thm:standard}
The envy cycle elimination algorithm, where in each iteration, the selected agent picks her favorite item, computes in polynomial time an allocation that is both EF1 and $1/2$-EFX for additive valuations.
\end{theorem}

In this section, we first adapt the ECE algorithm so that we can enforce that all agents receive bundles of size exactly $k$. Essentially, this is done by running the algorithm as usual and kicking agents out when they are allocated $k$ items. This version is presented as Algorithm~\ref{alg:ece-k}. 

\begin{algorithm}[tbh]
\caption{Envy cycle elimination algorithm for allocating bundles of size $k$}\label{alg:ece-k}
\begin{algorithmic}[1] 
\Procedure{EnvyCycleElimination}{$N,M$}

\For{$i \in N$}
    \State {$A_i \gets \emptyset$}
\EndFor

\While{there exist unallocated items}
    \State Let $G$ be the envy graph, i.e., $G = (N, \{(i, j) \mid v_i(A_i) < v_i(A_j)\})$
    \State $j \gets$ \Call{FindUnenviedAgent}{$G$} \Comment {Ties break in favor of smaller bundles}
    \State {Agent $j$ adds to $A_j$ her favorite unallocated item}
    \If {$|A_j| = k$} 
        \State {Agent $j$ becomes inactive} \Comment {Do not consider $j$ again}
    \EndIf
    \If {the envy graph has no source vertex}
        \State{Remove envy cycles till a source vertex is created}
    \EndIf
\EndWhile

\State \Return {$(A_1, A_2,\ldots,A_n)$}
\EndProcedure
\end{algorithmic}
\end{algorithm}

We exhibit that the performance of envy cycle elimination w.r.t. the \efx criterion is significantly different from its performance w.r.t. the EFX notion. 
In particular, in stark contrast to Theorem~\ref{thm:standard}, we show that for general additive valuations, Algorithm~\ref{alg:ece-k} cannot guarantee any approximation. 
Before we proceed, it is helpful to state the following simple yet useful facts.
The first one is a well-known property of envy cycle elimination and we include it for the sake of completeness.

\begin{lemma}[\cite{lipton2004approximately}] \label{lemma: agents-improve}
The value of an agent for the bundle allocated to her at the end of each iteration of Algorithm~\ref{alg:ece-k}, can only increase or remain the same.
\end{lemma}

\begin{lemma} \label{lemma: ef-implies-efx}
    For any $\gamma \in [0,1]$, if an allocation 
    $\aaa =  (A_1,\ldots, A_n)$ is $\gamma$-EF, then it is also $\gamma$-\efx.  
\end{lemma}

\begin{proof}
    Let $\aaa = (A_1, A_2, \ldots, A_n)$ be a $\gamma$-EF allocation for some $\gamma \in [0, 1]$, which means that $v_i(A_i) \geq \gamma \cdot v_i(A_j)$. 
    Consider any pair of agents $i$ and $j$, and suppose that agent $i$ envies agent $j$. If not, then $i$ trivially satisfies the $\gamma$-\efx condition w.r.t. agent $j$. 
    Due to the envy, there must exist at least one rational flip involving $a \in A_i$ and $b \in A_j$, such that $v_i(b) > v_i(a)$. Then
    \begin{align}
        v_i(A_i \pplus  b \mminus  a) > v_i(A_i) \geq \gamma \cdot v_i(A_j) > \gamma \cdot v_i(A_j \pplus  a \mminus  b). \nonumber
    \end{align}
    Hence, $\aaa$ is also a $\gamma$-\efx allocation.
\end{proof}

\begin{theorem}
\label{thm:ece-negative}
For any $\gamma\in(0, 1]$, Algorithm~\ref{alg:ece-k} fails to guarantee a $\gamma$-\efx allocation, even for $n=3$ agents.
\end{theorem}

\begin{proof}
    Consider an instance with $3$ agents and $3k$ items, denoted as $g_1,...,g_{3k}$. Agent $1$ has the following valuation function: $v_{1}(g_1)=1$ and $v_1(g_x)=\epsilon$ for $x = 2,\ldots,3k$, where $\epsilon < \frac{1}{k-1}$. 
    
    Agent $2$ has the following valuation function: $v_2(g_1)=C$ where $C>1$ is a large constant, $v_2(g_2)=1$,  $v_2(g_3) = \frac{1}{k-1} + \epsilon$, $v_2(g_x) = \frac{1}{k-1}$ for $x=4,\ldots,k+1$ and $v_2(g_x) = 1$ for $x=k+2,\ldots,3k$.
    
    Finally, agent $3$ has the following valuation function: $v_3(g_1)=v_3(g_2)=1$, $v_3(g_x) = \frac{1}{k-1}$ for $x=3,\ldots,k$, $v_3(g_{k+1}) = \frac{1}{k-1} - \epsilon$ and $v_3(g_{k+2}) = \epsilon$, where $\epsilon < \frac{1}{k-1}$. Note that the valuation of agent $3$ for items $g_{k+3},\ldots,g_{3k}$ can be arbitrary. 

    Consider a run of Algorithm~\ref{alg:ece-k}, where each iteration of the while-loop is referred to as a round. Initially, all agents are unenvied. Suppose agent $1$ is selected to receive the first item, $g_1$, followed by agent $2$, who chooses the second item, $g_2$. Agent $3$ is then the unique unenvied agent and selects item $g_3$. Note that agent $3$ remains unenvied until she selects all items up to $g_{k+1}$. Let $\mathcal{B} = (B_1, B_2, B_3)$ denote the partial allocation at the end of round $k+1$, before any envy cycles are resolved. At this point, we have $B_1 = \{g_1\}$, $B_2 = \{g_2\}$, and $B_3 = \{g_3, \ldots, g_{k+1}\}$. Figure \ref{figures: inapprox_general_fig1} depicts the corresponding envy graph $G$ for the partial allocation $\mathcal{B}$. Notice that there is no source vertex (i.e., no unenvied agent). But we can get one by removing the envy cycle between 2 and 3, according to the algorithm. Since at this point $v_2(B_2)=1 < 1 + \epsilon = v_2(B_3)$ and $v_3(B_3)=1 -\epsilon< 1 = v_3(B_2)$, agents $2$ and $3$ will exchange their bundles. Figure \ref{figures: inapprox_general_fig2} demonstrates the updated envy graph $G'$ corresponding to the new allocation $\mathcal{B}' = (B_1, B_3, B_2)$, which is depicted with the shaded area in Table~\ref{tab:ECEfailsGeneral}.

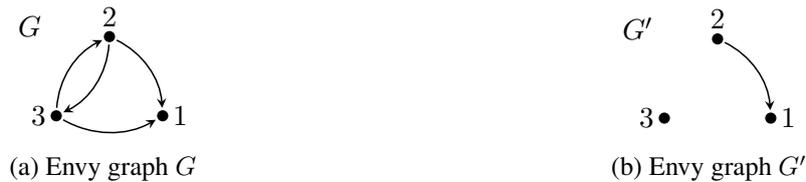
\begin{figure}[thbp]
\centering
\begin{subfigure}{0.49\columnwidth}
  \centering
  \begin{tikzpicture}[scale=0.7]
  \coordinate (i_3) at (0,0); 
  \coordinate (i_1) at (2,0);   
  \coordinate (i_2) at (1,1.5);     
  \node at (-0.5, 1.7) {$G$}; 
  
  \fill (i_3) circle (3pt) node[left] {$3$};
  \fill (i_1) circle (3pt) node[right] {$1$};
  \fill (i_2) circle (3pt) node[above] {$2$};
  
  \draw[->, bend left=30, >=stealth, line width=0.2mm, shorten >=3pt, shorten <=3pt] (i_2) to (i_3);
  \draw[->, bend left=30, >=stealth, line width=0.2mm, shorten >=3pt, shorten <=3pt] (i_3) to (i_2);
  \draw[->, bend right=30, >=stealth, line width=0.2mm, shorten >=3pt, shorten <=3pt] (i_3) to (i_1);
  \draw[->, bend left=30, >=stealth, line width=0.2mm, shorten >=3pt, shorten <=3pt] (i_2) to (i_1);
  \end{tikzpicture}
  \caption{Envy graph $G$}
  \label{figures: inapprox_general_fig1}
\end{subfigure}
\begin{subfigure}{0.5\columnwidth}
  \centering
  \begin{tikzpicture}[scale=0.7]
  \coordinate (i_3) at (0,0); 
  \coordinate (i_1) at (2,0);   
  \coordinate (i_2) at (1,1.5);     
  \node at (-0.5, 1.7) {$G'$}; 
  
  \fill (i_3) circle (3pt) node[left] {$3$};
  \fill (i_1) circle (3pt) node[right] {$1$};
  \fill (i_2) circle (3pt) node[above] {$2$};
  
  \draw[->, bend left=30, >=stealth, line width=0.2mm, shorten >=3pt, shorten <=3pt] (i_2) to (i_1);
  \end{tikzpicture}
  \caption{Envy graph $G'$}
  \label{fig:envy_graph_G_prime}
\end{subfigure}
\caption{Envy graphs for the proof of Theorem~\ref{thm:ece-negative}.}
\label{figures: inapprox_general_fig2}
\end{figure}

For the next step of the algorithm, assume that among the unenvied agents (i.e., the source vertices in Figure \ref{fig:envy_graph_G_prime}), agent $2$ is selected to receive the next item, $g_{k+2}$. Upon receiving this item and completing a bundle of $k$ items, agent $2$ becomes inactive. Following this, agents $1$ and $3$ each select $k-1$ additional items, all of which succeed item $g_{k+2}$ in Table~\ref{tab:ECEfailsGeneral} and hold a value of $1$ for agent $2$. Let $\aaa = (A_1, A_2, A_3)$ denote the final allocation. In this case, $v_2(A_2) = 2 + \epsilon$, and either $v_2(A_1) = k-1 + C$ or $v_2(A_3) = k-1 + C$, depending on how agents $1$ and $3$ share the items beyond $g_{k+2}$. W.l.o.g., assume that $v_2(A_1)=k-1 + C$. Observe that in this case there exists the rational flip $(g_3, g_x),\; x\geq k+3$, w.r.t. agent 2, where $g_3 \in A_2$, $g_x \in A_1$ and $v_2(g_x) > v_2(g_3)$, such that 
\begin{align}
    v_2(A_2 \pplus g_x \mminus g_3) = (2 + \epsilon) - \left(\frac{1}{k-1} + \epsilon\right) +1 = 3 - \frac{1}{k-1} < 3,
\end{align}
while
\begin{align}
    v_2(A_1 \mminus g_x \pplus g_3) = (k-1 + C) - 1 + \left(\frac{1}{k-1} + \epsilon\right) > C.
\end{align}

Therefore, we have established that Algorithm~\ref{alg:ece-k} cannot guarantee better than $3/C$-\efx allocations. Since $C$ is a parameter that we can make arbitrarily large, this means that we cannot have any approximation guarantee, for any $\gamma\in (0, 1]$.

This instance can be extended to an arbitrary number of agents $n$ by introducing multiple copies of agent $3$. For each such copy, we add a corresponding item that only these agents value at $1$. These agents will always receive an item before agent $2$,  and they will get any new item after agent  $2$.

\begin{table}[tbh]
\centering
    \renewcommand{\arraystretch}{1.3}
\begin{tabular}{c | c c c c c c c}
    & $g_1$ & $g_2$ &$g_3$ & $g_4$ & $\cdots$ & $g_{k+1}$  & $g_{k+2}$ \\

    \hline
    $v_1(\cdot)$ & \cellshade 1 & $\epsilon$ & $\epsilon$ & $\epsilon$ & $\cdots$  & $\epsilon$ & $\epsilon$   \\
        
    $v_2(\cdot)$ & $C$ & 1 &  \cellshade $\frac{1}{k-1} + \epsilon$ & \cellshade $\frac{1}{k-1}$ & \cellshade $\cdots$ &  \cellshade $\frac{1}{k-1}$  & $1$   \\
    
    $v_3(\cdot)$ & $1$ & \cellshade $1$ & $\frac{1}{k-1} $ & $\frac{1}{k-1}$ & $\cdots$ & $\frac{1}{k-1} - \epsilon$ & $\epsilon$  \\
\end{tabular}
\caption{The valuations for the proof of Theorem~\ref{thm:ece-negative} for the first $k+2$ items.}
\label{tab:ECEfailsGeneral}
\end{table}

\end{proof} 

Given the negative result for Algorithm~\ref{alg:ece-k} under general additive valuations, in the next subsections we focus on certain special cases that have been of interest in the fair division literature.

\subsection{Ordered valuations}

We start with a natural special case, where all agents agree on the ranking of the items, from the most valuable to the least valuable one. Such instances are often referred to as instances with ordered valuations, and have been often studied within fair division, e.g., \cite{bouveret2016characterizing,plaut2020almost}.
We show that Algorithm~\ref{alg:ece-k} does achieve a constant approximation in the following theorem. We again note the difference with the EFX notion, where, as proved in \cite{plaut2020almost}, the envy cycle elimination finds an exact EFX allocation.

\begin{theorem}
\label{thm:1/2-common ranking}
    For additive valuations and when all agents agree on the ranking of the items in terms of their value, Algorithm~\ref{alg:ece-k} returns a $1/2$-\efx allocation.
\end{theorem}

\begin{proof}
    The case $k\leq 2$ is covered by Theorem~\ref{theorem: existence-k-2} because under the common ranking assumption, and with $k=1$ or $k=2$, Algorithm~\ref{alg:ece-k} boils down to the generalized Round-Robin algorithm described in Theorem~\ref{theorem: existence-k-2}.
     
    Consider the case of $k \geq 3$. Let $A_i$ and $A_j$ denote the final bundles assigned to agents $i$ and $j$, respectively, at the end of the algorithm. Note that $|A_i| = |A_j| = k$. We will examine the envy of agent $i$ towards $j$. We will show that either (a) $i$ is $1/2$-EF w.r.t. $j$ or (b) that $i$ satisfies the \efx condition w.r.t. $j$. 

    Suppose that $v_i(A_i) < v_i(A_j)$, otherwise, we are done. We say that an agent is \emph{active} at a given iteration of the algorithm, if she still has not received $k$ items, and therefore she is still considered by the algorithm. Otherwise, once the agent is allocated her $k$-th item, she is considered \emph{inactive}. At this point, we distinguish between two cases.\\

\noindent    \textbf{Case 1: }\textit{Agent $i$ became inactive after agent $j$.}

    Let $g_k$ be the last item added to the bundle $A_j$ by the algorithm. After the allocation of $g_k$, agent $j$ became inactive, since she received $k$ items. Furthermore, since the owner of the bundle $A_j \mminus  g_k$ was unenvied just before $g_k$ was allocated, it follows by Lemma~\ref{lemma: agents-improve} that
    \begin{align}\label{thm:commonranking:ineq1}
        v_i(A_i) \geq v_i(A_j \mminus  g_k).
    \end{align}
    Note now that $|A_j \setminus g_k| = k-1 \geq 2$, and due to the common ranking assumption, all the items in $A_j \setminus g_k$ are at least as preferred as $g_k$, for all agents. Consequently, we have the following inequalities:
\begin{align}\label{thm:commonranking:ineq2}
        v_i(A_i) \geq v_i(A_j \setminus g_k) \geq (k-1) \cdot v_i(g_k).
    \end{align}
    By multiplying $k-1$ times inequality (\ref{thm:commonranking:ineq1}) and adding inequality (\ref{thm:commonranking:ineq2}), we obtain,
    \begin{align*}
        v_i(A_i) \geq \left(\frac{k-1}{k}\right) (v_i(A_j \setminus g_k) + v_i(g_k)) \geq \left(\frac{k-1}{k}\right)v_i(A_j).
    \end{align*}
    Since $k \geq 3$, this guarantees a $2/3$-EF approximation. By Lemma~\ref{lemma: ef-implies-efx}, this also satisfies the $2/3$-\efx condition, which is even better than the $1/2$-bound claimed in our theorem.
    
    We now turn to the complementary and more involved case, where the best we can get is a $1/2$-approximation.\\
    
    \noindent \textbf{Case 2:} \textit{Agent $i$ became inactive before agent $j$.}
    
    Let $g$ be the last item added to bundle $A_j$ before agent $i$ became inactive. We partition $A_j$ into three subsets: the set $X_j$, which contains all items added to $A_j$ before $g$; $X_j'$, containing items added after $g$; and the singleton set $\{g\}$. Similarly, it is convenient to also partition $A_i$ into two subsets: the set $X_i$, containing all items added to $A_i$ before the iteration where $g$ was added to $A_j$, and $X_i'$, consisting of all items added afterwards.
    
    Now, let $B_i$ represent the bundle held by agent $i$ at the time $g$ was allocated. It is important to note that $B_i$ may differ from $X_i$, as the bundle could have been transferred to another agent due to the resolution of an envy cycle. Nonetheless, by Lemma~\ref{lemma: agents-improve}, we know that $v_i(A_i) \geq v_i(B_i)$.
  Since the owner of the bundle $X_j$ was unenvied just before $g$ was allocated, we have the following:
\begin{align}\label{thm:commonranking:ineq3}
        v_i(A_i) \geq v_i(X_j). 
    \end{align}

    \noindent Furthermore, due to the common ranking assumption, \begin{align}\label{thm:commonranking:ineq4}
        v_i(h) \geq v_i(g) \geq v_i(g'), \text{ for all } h \in X_i \text{ and all } g' \in X'_j,
    \end{align}
    since every item $h \in X_i$ has been allocated before $g$, whereas every item $g' \in X'_j$ has been allocated after $g$.
    Similarly, \begin{align}\label{thm:commonranking:ineq5}
         v_i(h') \geq v_i(g') \text{ for all } h' \in X'_i \text{ and all } g' \in X'_j, 
    \end{align}
   \noindent since each item $h' \in X'_i$ was allocated while agent $i$ was still active, whereas every item $g' \in X'_j$ was allocated after agent $i$ became inactive (recall that $g$ was the last item added to $A_j$ before $i$ became inactive). At this point, we divide the analysis into two cases, based on the size of the set $X_i$.
   
   \noindent {\it Subcase 2a.} Suppose first that $X_i \neq \emptyset$, and let $h^* \in X_i$. Then we know that $v_i(h^*) \geq v_i(g)$ by \eqref{thm:commonranking:ineq4}. Since $|A_i \setminus h^*|=|X_i \cup X_i' \setminus h^*| = k-1$ and $|X_j'| \leq k-1$, by combining (\ref{thm:commonranking:ineq4}) and (\ref{thm:commonranking:ineq5}), we obtain 
   \begin{align*}
       v_i(A_i\setminus h^*) \geq v_i(X'_j).
   \end{align*}

    \noindent In conjunction with $v_i(h^*) \geq v_i(g)$, this implies that
    \begin{align}\label{thm:commonranking:ineq6}
        v_i(A_i) \geq v_i(X_j' \cup g).
    \end{align}
   
   \noindent Adding inequality (\ref{thm:commonranking:ineq3}) with (\ref{thm:commonranking:ineq6}) yields
   \begin{align}
        v_i(A_i) \geq \frac{1}{2} (v_i(X_j) + v_i(X'_j \pplus g))  = \frac{1}{2} v_i(A_j),
   \end{align}
   which guarantees that $i$ is $1/2$-EF w.r.t. agent $j$. By Lemma~\ref{lemma: ef-implies-efx}, this is also $1/2$-\efx.

   \noindent {\it Subcase 2b.} Otherwise, suppose $X_i = \emptyset$. Then $X_j = \emptyset$, otherwise $j$ would receive $g$ as (at least) her second item before $i$ receives any item, which cannot hold. This implies that $g$ is the highest-valued item in $A_j$ for all agents. 
   Let $h^1_i$ denote the highest valued item in $A_i$. Notice that in this case, $A_i = X'_i$ and $A_j = g \pplus X'_j$.
   The key observation in this case is that any rational flip in the final allocation involving agent $i$ w.r.t. $j$, must involve the item $g$ from the bundle $A_j$. Let $(r, g)$ represent any such rational flip, for some $r \in A_i$, where $v_i(g) > v_i(r)$. 
   Combining this with inequality (\ref{thm:commonranking:ineq5}), and since $A_i = X'_i$, $A_j = g \pplus X'_j$, it follows that
   \begin{align*}
       v_i(h) \geq v_i(g') \text{ for all } h \in A_i \pplus g \mminus r \text{ and all } g'\in A_j \mminus g \pplus r,
   \end{align*}
   which implies that
   \begin{align*}
       v_i(A_i \pplus g \mminus r) \geq v_i(A_j \mminus g \pplus r).
   \end{align*}
   Thus, this guarantees that agent $i$ satisfies the exact \efx condition w.r.t. agent $j$.

In conclusion, under all scenarios we examined, for any agents $i, j$, agent $i$ is either $2/3$-EF w.r.t. $j$ (Case~1) or $1/2$-EF w.r.t. $j$ (Subcase 2a) or \efx w.r.t. $j$ (Subcase 2b).  
Thus the $1/2$-\efx condition is satisfied. 
\end{proof}

We complement the previous result with an example that shows a lower bound on the performance of Algorithm~\ref{alg:ece-k}. Although the example is not tight, it does not leave too much room for improvement either (in fact it is tight w.r.t. Case 1 within the proof of Theorem~\ref{thm:1/2-common ranking}).

\begin{theorem}
\label{exm:2/3lb}
Algorithm~\ref{alg:ece-k} cannot guarantee a better than a $2/3$-\efx approximation for additive valuations, when the agents agree on the ranking of the items w.r.t. their value.
\end{theorem}

\begin{proof}
    We use an example with 6 items and 2 agents.
    The table below shows the valuations for the items $\{g_z\}_{z \in [6]}$ for agents $1, 2$, where $0 < \epsilon \ll 1$ is arbitrarily small. We assume the agents have identical valuations, denoted as $v (\cdot)$.

    \begin{center}
    \begin{tabular}{l|c c c c c c c c c}
        & $g_1$ & $g_2$ &$g_3$ & $g_4$ & $g_5$ & $g_6$\\
        \hline
         $v(\cdot)$&$\cellshade 1 + 2\epsilon$ &$1 + \epsilon$  &$1$ &$\cellshade 1 - 2\epsilon$ &$\cellshade 1 - 3\epsilon$ &0\\
    \end{tabular}        
    \end{center}

    \noindent It is easy to verify that the allocation produced by Algorithm~\ref{alg:ece-k} for this instance consists of the bundles $\{g_1, g_4, g_5\}$ (shown shaded in the table above) and $\{g_2, g_3, g_6\}$. W.l.o.g., let $A_1 = \{g_2, g_3, g_6\}$ and $A_2 = \{g_1, g_4, g_5\}$. Then, $v(A_1) = 2 + \epsilon$ and $v(A_2) = 3 - 3\epsilon$. This implies that agent $1$ envies agent $2$. Observe that agent $1$ envies agent $2$ even after the rational flip $(g_2, g_1)$. If $A_{1}'$ and $A_{2}'$ are the bundles of agents $1$ and $2$, respectively, after performing the rational flip, then $v(A_{1}') = 2 + 2\epsilon $ and $v(A_{2}') = 3 - 4\epsilon$, 
    with the envy ratio approaching     \begin{align*}
        \frac{v(A_1')}{v(A_2')} =  \frac{2 + 2\epsilon}{3 - 4\epsilon} \approx \frac{2}{3}.
    \end{align*}
Hence, the algorithm cannot produce a better than a $(2/3 + \delta)$ approximation for any $\delta>0$.  
\end{proof}

\subsection{Agreement on the top $n$ items}
\label{subsec:top-n}

In this section, we relax further the common ranking assumption on the items. We feel that an even more natural scenario than the  case of ordered valuations is that the agents only agree on what are the most valuable items. In particular, we will focus on instances where the set of the $n$ most valuable items is the same for all agents, without having to agree on the ranking of these $n$ items. 
We denote by $T^n$ the set of these top items. 
This setting was recently studied in \cite{markakis23improved}.

Given the negative results of Theorem~\ref{thm:ece-negative}, Algorithm~\ref{alg:ece-k} is not expected to perform well. To overcome this obstacle, we modify the envy cycle elimination procedure, so that it becomes more appropriate for our constrained model. The main problem with Algorithm~\ref{alg:ece-k} is that as soon as an agent acquires a bundle of size $k$, then she is essentially removed from the process and she is not considered again for receiving any additional items. This can create a problem if in the subsequent iterations, the other agents receive items that may create rational flips. In order to make the algorithm more robust and more tailored to our setting, we introduce the following modifications.

\begin{itemize}[noitemsep,leftmargin=*]
    \item \emph{Item swap operations}. We allow certain agents who have already acquired a bundle of size $k$ to swap their least valued item from their bundle with their  most preferred item from the pool of currently unallocated items, provided the swap strictly improves their bundle.
    \item \emph{Privileged agents}. We maintain a set of privileged agents, $P$. For an agent $i$ to enter the set $P$ it must hold that $i$ already has a bundle of size $k$, she is unenvied by agents in $N\setminus P$, and she cannot currently improve her allocation by an item swap.  Such agents remain in $P$ until they find a profitable swap (with items that may become available in the sequel). The algorithm also maintains the invariant that there is no envy from an agent in $N\setminus P$ to an agent in $P$.  
\end{itemize}

The algorithm is described as Algorithm~\ref{alg: algorithm_general}. To briefly describe its main steps, the main intuition is that in each iteration we give priority to the privileged agents so as to avoid creating detrimental flips for them. In particular, in the beginning of each iteration, every privileged agent $j$, examined in a certain order, is given the chance to do an item swap. If this succeeds, $j$ stops being privileged and is kicked out of $P$ along with all other reachable nodes in the updated envy graph, when restricted to $P$. 
If, on the other hand, no privileged agent had any beneficial swap, then the algorithm finds an unenvied agent $i$ in $N\setminus P$, by using envy cycle elimination on the envy graph restricted to $N\setminus P$. If $i$ already has $k$ items and does not have any beneficial swap, then $i$ is added to $P$. Otherwise, either $i$ performs an item swap or she has less than $k$ items and receives her most preferred item from the unallocated ones. Figure~\ref{fig:graph_of_general_approx_efx} provides an illustrative example of the envy graph during the algorithm’s execution.

To see why the algorithm needs to give first priority to the privileged agents, suppose that it did not do so. Then the agents of $P$ may already envy some other agent $j\in N\setminus P$ in the beginning of the current iteration. If $j$ is selected to receive the next item, without first checking that the agents in $P$ do not have a profitable item swap, this can definitely create bad flips that would violate the (approximate) EFFX condition.

\begin{algorithm}[h!t]
\caption{Envy cycle elimination algorithm with swaps for allocating bundles of size $k$}
\begin{algorithmic}[1]
\Procedure{EnvyCycleEliminationWithSwaps}{$N, M$}

\For{$i \in N$}
    \State $A_i \gets \emptyset$ 
\EndFor

\State $S \gets M$ \Comment{Initialize pool of available items}
\State $P \gets \emptyset$ \Comment{Set of privileged agents}

\While{$S \neq \emptyset$}
    \State If $P\neq \emptyset$, build envy graph $G_P = (P, \{(i, j) \mid i, j \in P, \; v_i(A_i) < v_i(A_j)\})$
    
    \State flag $\gets$ False \Comment{Tracks whether a privileged agent left $P$}

    \For{$p$ in topological order of $G_P$}
        \State $g \gets \operatorname{argmin}_{g' \in A_p} v_p(g')$ \Comment{Least valued item owned by $p$}
        \State $g^* \gets \operatorname{argmax}_{g' \in S} v_p(g')$ \Comment{Most valued item in pool $S$}

        \If{$v_p(g) < v_p(g^*)$}
            \State $A_p \gets A_p \setminus \{g\} \cup \{g^*\}$
            \State $S \gets S \setminus \{g^*\} \cup \{g\}$
            \Comment{Agent $p$ swaps her least preferred item}
            \State flag $\gets$ True

            \State Update $G_P$; remove from $P$ node $p$ and any node reachable from $p$
            \State \textbf{break}
        \EndIf
    \EndFor

    \If{not(flag)}
        \State Build envy graph $G_{N \setminus P} = (N \setminus P, \{(i, j) \mid v_i(A_i) < v_i(A_j)\})$

        \State $j \gets$ \Call{FindUnenviedAgent}{$G_{N \setminus P}$}
        \Statex \Comment{Agent $j$  is \emph{locally unenvied} in $N\setminus P$; ties break in favor of smaller bundles}

        \State $g^* \gets \operatorname{argmax}_{g' \in S} v_j(g')$

        \If{$|A_j| < k$}
            \State $A_j \gets A_j \cup \{g^*\}$
            \State $S \gets S \setminus \{g^*\}$
            \Comment{Agent $j$ \emph{gets} one item if bundle not full}
        \Else
            \State $g \gets \operatorname{argmin}_{g' \in A_j} v_j(g')$
            \If{$v_j(g) < v_j(g^*)$}
                \State $A_j \gets A_j \setminus \{g\} \cup \{g^*\}$
                \State $S \gets S \setminus \{g^*\} \cup \{g\}$
                \Comment{Agent $j$ \emph{swaps} to improve her bundle}
            \Else
                \State $P \gets P \cup \{j\}$
                \Comment{Agent $j$ \emph{passes} and becomes privileged}
            \EndIf
        \EndIf
    \EndIf
\EndWhile
\State \Return{$(A_1,A_2,\ldots, A_n)$}
\EndProcedure
\end{algorithmic}
\label{alg: algorithm_general}
\end{algorithm}

To state the main result of this subsection, we need to introduce one more parameter. Given the set $T^n$ of the common top $n$ items, let $\rho$ be the maximum ratio between the most valuable and the least valuable item in $T^n$, where the maximum is taken over all agents, i.e., $\rho = \max_{i \in [n]} \max_{g, g'\in T^n}\frac{v_i(g)}{v_i(g')}$.

\begin{theorem}
    \label{thm: statement_of_common_top_n}
    Algorithm~\ref{alg: algorithm_general} runs in polynomial time and for $k\geq 2$, if all agents agree on the set of top-$n$ items, it returns a $\min\{1/3, 1/(\rho+1)\}$-EF allocation. More precisely, for any pair of agents $i, j$, either $i$ satisfies $1/(\rho+1)$-EF and \efOne w.r.t. $j$ or $1/3$-EF w.r.t. $j$. For $k=1$, it returns an \efx allocation.
\end{theorem}

The remainder of Section~\ref{subsec:top-n} is dedicated to the proof of Theorem~\ref{thm: statement_of_common_top_n}. 
We first present some useful lemmas that will guide us for the proof.
The following lemma captures an easy but important observation, concerning the invariant that the algorithm maintains throughout its execution.

\begin{lemma}
\label{lem:invariant}
    An invariant of the algorithm is that at the end of each iteration, there is no envy from an agent in $N\setminus P$ towards an agent in $P$. Furthermore, at any point during the algorithm, if an agent~$j$ is allowed to select an item when agent~$i$ envies her, then it must be that $i \in P$.
    \label{lemma: common_top_n_local_envy}
\end{lemma}

\begin{proof}
    The invariant is easy to see by the construction of the algorithm and by induction. Initially $P=\emptyset$, hence there is nothing to prove. During the algorithm, an agent enters $P$ only if she is unenvied by $N\setminus P$. Furthermore, while in $P$, if she performs an item swap, and this may create envy from $N\setminus P$, she is removed from $P$, along with all other agents that are reachable from her in $P$ (in the updated envy graph $G_P$). This way, we maintain that there is no envy from $N\setminus P$ towards $P$.  
    
    To see the second statement, suppose agent~$i$ envies agent~$j$.
    If $j\in P$, then by the invariant, we must have that $i\in P$. If  $j\not\in P$, then $j$ became eligible to select an item, because she was locally unenvied in $G_{N\setminus P}$. Therefore again, it must hold that $i \in P$.
\end{proof}

Before stating the next lemmas, we introduce some notation. For each agent $i \in [n]$, let $g_i^\ell$ denote her $\ell$-th most preferred item for some $\ell \in [kn]$, i.e., the item that ranks $\ell$-th in her preference ordering over all items. The proof of the first lemma is deferred to Appendix~\ref{app:modified_ece}.

\begin{lemma}
    Let $i \in [n]$, and let $B_i \neq \emptyset$ be a bundle held by agent~$i$ at some point after the first $n$ items have been allocated. Then it holds that $v_i(B_i) \geq v_i(g_i^n)$. Moreover, when all agents agree on the set of top-$n$ items, no two such items ever appear in the same bundle.
    \label{lemma: common_top_n_min_value}
\end{lemma}

\begin{figure}[tbh]
  \centering
  \begin{tikzpicture}[scale=1, >=Stealth,
      every node/.style={circle, fill=black, minimum size=4.5pt, inner sep=0pt},
      every path/.style={->, line width=0.2mm, shorten >=1pt, shorten <=1pt}
  ]

  \begin{scope}[on background layer]
      \fill[customgray, opacity=1.0] (4.2,2.3) ellipse (4.5 and 0.8);
  \end{scope}

  \node (p1) at (0.7,2.3) {};
  \node (p2) at (1.7,2.3) {};
  \node (p3) at (2.7,2.3) {};
  \node (p4) at (4.3,2.3) {};
  \node (p5) at (5.3,2.3) {};
  \node (p6) at (6.9,2.3) {};
  \node (p7) at (7.9,2.3) {};

  \draw (p1) to (p2);
  \draw (p2) to (p3);
  \draw[->, line width=0.2mm, shorten >=3pt, shorten <=3pt] 
    (p3) -- node[midway, fill=customgray, inner sep=1pt] {\small $\cdots$} (p4);
  \draw (p4) to (p5);

  \draw[->, line width=0.2mm, shorten >=3pt, shorten <=3pt]
    (p5) -- node[midway, fill=customgray, inner sep=1pt] {\small $\cdots$} (p6);

  \draw (p6) to (p7);

  \draw[bend left=20] (p2) to (p5);
  \draw[bend right=20] (p1) to (p4);

  \begin{scope}[xshift=-1cm]

    \node (a1) at (1.6, 0.0) {};
    \node[draw=none, fill=none, inner sep=0pt] at (1.5, -0.3) {$p_1$};
    \node (a2) at (2.3, -0.8) {};
    \node (a3) at (3.2, -0.2) {};
    \node (a4) at (3.0, 0.8) {};
    \node (a5) at (4.2, 0.45) {};
    \node (a6) at (4.6, -0.4) {};
    \node (a7) at (3.8, -1.3) {};

    \node[fill=white, inner sep=1pt] (dotsB) at (3.8, -0.2) {\small $\cdots$};

    \draw[->, dotted, bend left=15] (a2) to (a3);
    \draw[->, dotted, bend left=15] (a3) to (a4);
    \draw[bend left=15] (a1) to (a4);
    \draw[->, dotted, bend left=15] (a4) to (a5);
    \draw[->, dotted, bend left=15] (a5) to (a6);
    \draw[->, dotted, bend left=15] (a6) to (a7);
    \draw[->, dotted, bend left=15] (a7) to (a2);

    \node[fill=white, inner sep=1pt] (dotsA) at (5, -0.2) {\small $\cdots$};

    \node (b1) at (5.6, 0.0) {};
    \node[draw=none, fill=none, inner sep=0pt] at (5.5, -0.3) {$p_2$};
    \node (b2) at (6.6, 0.6) {};
    \node (b3) at (7.3, -0.2) {};
    \node (b4) at (7.8, -0.6) {};
    \node (b5) at (6.8, -1.1) {};
    \node (b6) at (8.4, -1.3) {};
    \node (b7) at (8.8, -0.1) {};
    \node (b8) at (8.0, 0.7) {};

    \node[fill=white, inner sep=1pt] (dotsB) at (9.3, -0.2) {\small $\cdots$};

    \draw[bend left=15] (b4) to (b5);
    \draw[bend left=15] (b5) to (b1);
    \draw[bend right=15] (b2) to (b5);
    \draw[bend right=15] (b3) to (b2);
    \draw[bend left=15] (b3) to (b8);

    \draw[bend right=15] (b6) to (b7);
    \draw[bend right=15] (b7) to (b8);
  \end{scope}

  \draw[dashed, bend right=15] (p3) to (a1);
  \draw[dashed, bend right=15] (p6) to (b1);

  \node[draw=none, fill=none, left] at (0.2,3.2) { $G_P$};
  \node[draw=none, fill=none, left] at (0.6,1) { $G_{N \setminus P}$};

  \end{tikzpicture}
  \caption{Illustration of the envy graph during an execution of Algorithm~\ref{alg: algorithm_general}.
    The set of privileged agents \( P \) corresponds to the nodes of the DAG \( G_P \), shown at the top of the figure inside the gray ellipse. The edges of \( G_P \) represent envy among privileged agents and are shown as solid arrows. These edges follow a topological order, since \( G_P \) is acyclic by Lemma~\ref{lemma:DAG}. At the bottom, the agents in \( N \setminus P \) form the graph \( G_{N \setminus P} \), which may contain envy cycles—shown as dotted directed edges. Dashed arrows from \( G_P \) to \( G_{N \setminus P} \) represent envy from privileged agents toward non-privileged ones. These edges can exist (e.g., from nodes in \( G_P \) to \( p_1 \) and \( p_2 \)), but the reverse cannot occur, as established by Lemma~\ref{lem:invariant}. Agent \( p_1 \) is \emph{locally unenvied} in \( G_{N \setminus P} \), since she is only envied by agents in \( P \). In contrast, agent \( p_2 \) is not locally unenvied. As a result, \( p_1 \) may be selected in the next iteration, while \( p_2 \) will not.}
    \label{fig:graph_of_general_approx_efx}
\end{figure}
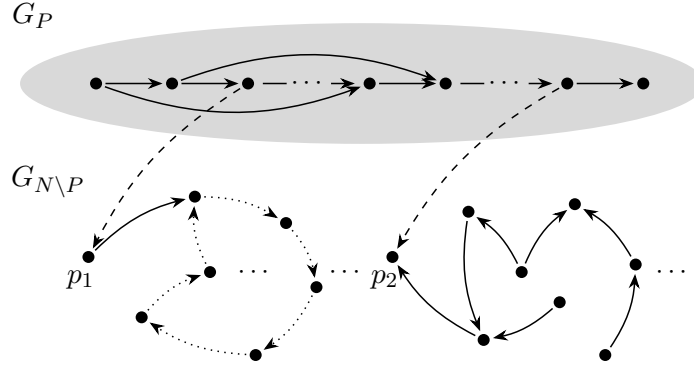

The following lemma guarantees that the algorithm can process the agents in set $P$ following their topological order in the graph $G_P$. 

\begin{lemma}\label{lemma:DAG}
    The envy graph $G_P$ is a directed acyclic graph (DAG) at every step of the algorithm.
\end{lemma}

\begin{proof}
Initially it is a DAG since it is an empty graph ($P=\emptyset$). Given the invariant from Lemma~\ref{lem:invariant}, that there is no envy from $N\setminus P$ to $P$, then any time we add a node from $N\setminus P$ to $P$, this node cannot participate in a cycle since it has no outgoing edges towards other nodes of $P$. Then, any time we remove a node from $P$ again, this maintains the DAG property.
\end{proof}

Finally, the last lemma shows that Algorithm~\ref{alg: algorithm_general} terminates in polynomial time. The proof is deferred to Appendix~\ref{app:modified_ece}. 

\begin{lemma}
\label{lem:polynomial}
    Algorithm~\ref{alg: algorithm_general} always terminates after a polynomial number of iterations.
\end{lemma}

\noindent We are now ready to prove our main theorem.

\begin{proof}[Proof of Theorem~\ref{thm: statement_of_common_top_n}]

Let $A_i$ and $A_j$ denote the final bundles assigned to agents $i$ and $j$, respectively, at the end of the algorithm. Note that $|A_i| = |A_j| = k$. 
For $k=1$, the proof is trivial. For $k\geq 2$, we fix a pair of agents $i, j$ and we will examine the envy of agent $i$ towards $j$. We will show that either (a) $i$ is $1/(\rho+1)$-EF and \efOne w.r.t. $j$ or (b) $i$ is $1/3$-EF w.r.t. $j$. 

Let $(g_r)_{r \in [k]}$ be the items of the finalized bundle $A_j$ ordered according to the time they were allocated to $A_j$, i.e., $g_r$ became part of $A_j$ before $g_{r+1}$ for all $r \in [k-1]$. We focus on the item $g_r \in A_j$ with the maximum index $r$ such that at the time that $g_r$ was added to $A_j$, agent $i$ did not envy the owner of $A_j$ (before the addition of $g_r$).
This means that for the goods $g_1,\dots, g_{r-1}$, agent $i$ may or may not envy the owner of $A_j$ at the time of their addition. But by the definition of $r$, $i$ envies the bundle $A_j$ right before the addition of $g_{r+1},\dots, g_k$.
We partition $A_j\setminus\{g_r\}$ into two sets: $X_j=\{g_1,...,g_{r-1}\}$, and $X'_j=\{g_{r+1},...,g_k\}$. Let $B_i$ denote the bundle held by agent $i$ at the time when $g_r$ was allocated.

Assume that $v_i(A_i) < v_i(A_j)$; otherwise, the claim holds trivially. At this point, we distinguish between two cases, based on whether $B_i$ is empty or not.\\

\noindent    \textbf{Case 1: }$B_i = \emptyset$. Then, by the algorithm’s tie-breaking rule (which favors agents with fewer items), agent $j$ must have received $g_r$ as her first item. Thus, $r=1$, and so $X_j = \emptyset$, $X_j' = \{g_2,\ldots, g_k\}$, and $A_j = \{g_r\} \cup X_j'$.
Next, due to Lemma~\ref{lemma: common_top_n_min_value}, and under the assumption that $v_i(g_i^n)/v_i(g_i^1) \geq 1/\rho$ we have that
\begin{equation}
   v_i(A_i) \geq v_i(g_i^n) \geq \frac{1}{\rho}   v_i(g_1).
   \label{alg: com-top-n: case1_ineq_1_true}
\end{equation}

Next, consider the allocation of items $g_{r+1}, \ldots, g_{k}$ (i.e., the items in $X_j'$). Observe that agent $i$ envies $A_j$ during the allocation of these items. Recall that $g_r$ was the last item allocated to $A_j$ before $i$ began to envy $A_j$. Thus, at the time each item $g_{r+\ell}$ is allocated, agent $i$ belongs\footnote{This does not imply that $i$ remains in $P$ continuously throughout the allocations of $g_{r+1}$ and $g_k$, but rather that $i$ is in $P$ precisely at the moment each of these items is allocated.} to set $P$, due to Lemma~\ref{lemma: common_top_n_local_envy}. Since $i \in P$, we have two cases:

\begin{itemize}
    \item[(1)] If the owner of $A_j$ belonged to $P$ as well, then (between the two of them) $i$ should have had higher priority to receive the item $g_{r+\ell}$, since $i$ appears earlier in the topological order of $G_P$.
    \item[(2)] If the owner of $A_j$ did not belong to $P$, then $i$ still had priority, as all agents in $P$ have priority over those not in $P$.
\end{itemize}

In both cases, $i$ had the opportunity to receive item $g_{r+\ell}$ but declined it. This must be because the least valuable item in $i$'s current bundle was worth more than $v_i(g_{r+\ell})$. Consequently, we have that in the end, 
$$
\frac{v_i(A_i)}{k} \geq v_i(g_{r+\ell}) \quad \text{for all } \ell \in \{1, \ldots, k - r\}.
$$
 Summing over all such $\ell$, we obtain

\begin{align}
    \frac{k-1}{k} v_i(A_i) \geq \frac{|X_j'|}{k}v_i(A_i) = \sum_{\ell=1}^{k-r} v_i(A_i)/k \geq  \sum_{\ell=1}^{k-r}v_i(g_{r+\ell}) = v_i(X_j').
    \label{alg: com-top-n: case1_ineq_2_true}
\end{align}

\noindent Combining inequalities (\ref{alg: com-top-n: case1_ineq_1_true}) and (\ref{alg: com-top-n: case1_ineq_2_true}), we conclude that

\begin{equation}
    v_i(A_i)\geq \frac{k}{(1+\rho)k  - 1} v_i(A_j).
    \label{alg: com-top-n: case1_ineq_3_true}
\end{equation}

\noindent Hence, in this case agent $i$ is $1/(\rho+1)$-EF w.r.t.$j$ (in fact, inequality (\ref{alg: com-top-n: case1_ineq_3_true}) provides a slightly stronger guarantee).

Moreover, in this case $i$ is also \efOne w.r.t. $j$. To see this, let $g$ be $i$'s least valuable item in $A_i$. Then, by an averaging argument, we have that 

\begin{align}
    v_i(g) \leq \frac{v_i(A_i)}{k}.
    \label{ineq: averaging_argument}
\end{align}

\noindent From inequality (\ref{alg: com-top-n: case1_ineq_2_true}) it follows that 
\begin{align*}
    \frac{k-1}{k}v_i(A_i) &\geq v_i(X_j')\\
    v_i(g_1) +  \frac{k-1}{k}v_i(A_i) &\geq v_i(A_j)\\
    v_i(g_1) - v_i(g)+ \frac{k-1}{k}v_i(A_i) &\geq v_i(A_j) - v_i(g_1)\\
    v_i(A_i \cup g_1 \setminus g) &\geq v_i(A_j \setminus g_1 \cup g), 
\end{align*}
where the second-to-last inequality holds because $v_i(g) < v_i(g_1)$ and the last inequality holds due to the averaging argument in (\ref{ineq: averaging_argument}). Therefore, the definition of \efOne is satisfied.\\

\noindent    \textbf{Case 2: }$B_i \neq \emptyset.$  At the time $g_r$ was allocated, agent $i$ did not envy the owner of $g_r$, so 
\begin{equation}
    v_i(B_i) \geq v_i(X_j).
    \label{alg: com-top-n: case2_ineq_1_true}
\end{equation}

\noindent    We now show that $v_i(B_i) \geq v_i(g_r)$. Suppose, for the sake of contradiction, that 
\begin{equation}
    v_i(B_i) < v_i(g_r)
    \label{alg: com-top-n: case2_ineq_2_false}
\end{equation}
    
Then, the owner of $B_i$ up until that point did not have access to $g_r$, i.e., $g_r$ was not in $S$ nor was any other item that is at least as valuable as $g_r$ when forming $B_i$, or else, by Lemma~\ref{lemma: agents-improve} she would have acquired an item of at least that value. This implies that agent $j$ received $g_r$ due to a swap from another agent $z$, who had held it earlier. Let $B_{z}$ be the bundle held by $z$ immediately prior to this swap, i.e., $g_r \in B_z$. Observe that, due to the algorithm, agent $z$ was allowed to swap $g_r$ because one of the following held:
\begin{enumerate}
    \item[(1)] Agent $z$ was unenvied by $i$ and hence $v_i(B_i) \geq v_i(B_{z}) \geq v_i(g_r)$,  contradicting inequality (\ref{alg: com-top-n: case2_ineq_2_false}).

    \item[(2)] Agent $z$ was envied by $i$ but locally unenvied. Hence, by Lemma~\ref{lemma: common_top_n_local_envy},  $i\in P$. Moreover, since $B_i \neq \emptyset$, Lemma~\ref{lemma: common_top_n_min_value} implies that $v_i(B_i) \geq v_i(g_i^n)$. Therefore, due to inequality (\ref{alg: com-top-n: case2_ineq_2_false}), it follows that $g_r \in \{g_i^1, \ldots, g_i^{n-1}\}$. Given the assumption that all agents agree on the set of top-$n$ items, $g_r \in \{g_z^1, \ldots, g_z^n\}$, and so, by Lemma~\ref{lemma: common_top_n_min_value}, $g_r$ is the top item of $B_z$. But this contradicts the swap rule: top items are never dropped in swaps—the dropped item is always the $k$-th (least preferred) item in one's bundle.
\end{enumerate}
Thus, in both subcases, we derive a contradiction, and conclude that $v_i(B_i) \geq v_i(g_r)$.
By Lemma~\ref{lemma: agents-improve}, it follows that

\begin{equation}
    v_i(A_i) \geq v_i(g_r).
    \label{alg: com-top-n: case2_ineq_2_true}
\end{equation}

Next, consider the allocation of items $g_{r+1}, \ldots, g_{k}$ (i.e., the items in $X_j'$). Using an argument similar to that in Case 1 for the same set of items (see inequality (\ref{alg: com-top-n: case1_ineq_2_true})), we obtain

\begin{align}
    v_i(A_i) \geq v_i(X_j').
    \label{alg: com-top-n: case2_ineq_3_true}
\end{align}

\noindent Combining inequalities (\ref{alg: com-top-n: case2_ineq_1_true}), (\ref{alg: com-top-n: case2_ineq_2_true}) and (\ref{alg: com-top-n: case2_ineq_3_true}), we conclude that $v_i(A_i)\geq \frac{1}{3}v_i(A_j)$, implying that agent $i$ is $1/3$-EF w.r.t. $j$.

To complete the analysis, it remains to establish that Algorithm~\ref{alg: algorithm_general} terminates in a polynomial number of steps, which follows from Lemma~\ref{lem:polynomial}.
\end{proof}

\subsection{Bounded ratio within the top $n$ items}
In this final part of Section~\ref{sec:ece}, we drop the earlier assumptions on agreements between agents on the ranking of the items. Instead, let $T_i^n$ denote the set of the top-$n$ items according to agent $i$'s valuation, for each $i \in [n]$. We focus on instances that are $\rho$-bounded with respect to these individualized sets $T_i^n$.

\begin{definition}[$\rho$-bounded instances w.r.t. $T_i^n$ for $i\in \text{[}n\text{]} $]
    Given a set of valuations $\{v_i(\cdot)\}_{i \in [n]}$ and $kn$ items, then $\rho= \max_{i \in [n]}  {v_i(g^1_i)}/{v_i(g_i^n)}$. 
\end{definition}

By exploiting the proof of Theorem~\ref{thm: statement_of_common_top_n}, we can have the following corollary showing that Algorithm~\ref{alg: algorithm_general} guarantees a $1/(\rho+2)$-EF allocation.

\begin{corollary}
    For $k\geq 2$, an allocation returned by Algorithm~\ref{alg: algorithm_general} is $1/(\rho+2)$-EF for $\rho$-bounded instances w.r.t. $T_i^n$ for $i\in[n]$. For $k=1$, the algorithm returns an \efx allocation.
\end{corollary}

\begin{proof}
    The proof proceeds similarly to the proof of Theorem~\ref{thm: statement_of_common_top_n}. The analysis for Case 1 remains unchanged and yields a $1/(\rho+1)$-EF guarantee.

    For Case 2, the only difference lies in inequality~(\ref{alg: com-top-n: case2_ineq_2_true}). Since we have dropped the assumption that agents agree on the set of top-$n$ items, we can no longer ensure that agent $z$ retains item $g_r$; that is, although $g_r \in \{g_i^1,\ldots, g_i^{n-1}\}$, it is not necessarily true that $g_r \in \{g_z^1,\ldots, g_z^{n}\}$.

    However, by Lemma~\ref{lemma: common_top_n_min_value}, we obtain the following bound:
    \begin{align}\label{alg: com-top-n: rho-bounded_ineq2_true}
        v_i(A_i) \geq v_i(g_i^n) \geq \frac{1}{\rho} v_i(g_i^1) \geq \frac{1}{\rho}v_i(g_r).
    \end{align}

    \noindent Combining inequalities~(\ref{alg: com-top-n: case2_ineq_1_true}),~(\ref{alg: com-top-n: rho-bounded_ineq2_true}), and~(\ref{alg: com-top-n: case2_ineq_3_true}), we conclude that in Case 2,
    \begin{equation}
        v_i(A_i) \geq \frac{1}{\rho+2} v_i(A_j).
    \end{equation}

    Therefore, the allocation returned is $1/(\rho+2)$-EF.
\end{proof}
\section{Approximate fairness and Pareto optimality}
\label{sec:PO}

In this section, we focus on the existence of $\gamma$-\efOne\ allocations that are also efficient. Our main result establishes that it is always possible to guarantee allocations that are both $1/2$-\efOne\ and Pareto optimal, by leveraging the well-known \emph{Maximum Nash Welfare} (MNW) solution, adapted to accommodate our cardinality constraints.  Pareto optimality (PO) is a fundamental efficiency criterion, ensuring that no reallocation can make all agents weakly better off and at least one agent strictly better off. Formally, an allocation $\aaa$ is Pareto optimal if there is no allocation $\mathcal{B}$ such that for all $i \in [n]$, $v_i(B_i) \geq v_i(A_i)$, and for some agent $j \in [n]$, $v_j(B_j) > v_j(A_j)$.

Among the Pareto optimal allocation rules, one stands out significantly in fair division allocations. The maximum Nash welfare solution already guarantees EF1 and PO solutions in the original unconstrained model~\cite{caragiannis2019unreasonable}. Under matroid constraints it can provide a $1/2$-EF1 and PO allocation~\cite{cookson2025constrained,wang2024fairness}, while under budget constraints it provides a $1/4$-EF1 and PO guarantee~\cite{WuLG25}.

The {maximum Nash welfare} rule selects an allocation of $k$-sized bundles that maximizes the product of values: $\aaa \in \arg\max \prod_{i=1}^n v_i(A_i)$. In the fringe cases\footnote{This is possible e.g., in the extreme case where $v_i(g)>0$ only of a single item $g \in [kn]$.} where every allocation yields a Nash welfare of $0$, the rule finds a maximal set $S \subset [n]$ of agents that can get a positive value with $k$-sized bundles, and subject to that, it maximizes the Nash welfare of the agents in $S$. For the sake of completeness, we formally defined the $k$-bundle version of the MNW using the similar definition of~\cite{cookson2025constrained} for generally constraint bundles. 

\begin{definition}[$k$-bundle maximum Nash welfare]
For an allocation $\aaa$ of $k$-sized bundles let $P(\aaa=\{i \in N: v_i(A_i)>0\}$.
An allocation $\mathcal{A}$ is a $k$-bundled maximum Nash welfare allocation if (1) it maximizes the number of agents receiving positive value $P(\aaa)$, and subject to that, (2) maximizes the product $\prod_{i \in P(\aaa)} v_i(A_i)$. 
\end{definition}

The main result of this section shows that MNW allocation rule always return allocations that are $1/2$-\efOne. Combined with the fact that MNW is PO, this establishes the existence of allocations that are both $1/2$-\efOne and PO. This is the first positive result for the \efOne and PO regime for general additive valuations. The proof follows an idea by~\cite{cookson2025constrained} used to prove the existence of $1/2$-EF1 and PO allocations for base-orderable matroid constraints.

\begin{theorem}\label{thm:PO:HalfEFF1}
   Any $k$-bundle MNW-optimal allocation is $1/2$-\efOne.
\end{theorem}

\begin{proof}
Let $\aaa=(A_1,...,A_n)$ be a $k$-bundled MNW allocation. This is that $\aaa$ yields the maximum positive Nash welfare if this is possible, or it maximizes the Nash welfare of all non-zero valued agents, subject on minimizing the number of zero valued agents.

Assume for the sake of contradiction that there exists a pair of agents $i,j$ such that (1) $i$ envies $j$, i.e., $v_i(A_i) < v_i(A_j)$, and (2) for all pairs $g^* \in A_j$, $g \in A_i$ such that $v_i(g^*) > v_i(g)$: 
\begin{align}\label{eq:EFF1_NASH_1_contradiction}
    v_i(A_i) + v_i(g^*) -v_i(g) < \frac{1}{2}\left(v_i(A_j) +v_i(g) - v_i(g^*) \right)
\end{align}

We will focus now on the rational flips between the two bundles $A_i$ and $A_j$, from the perspective of the envious agent $i$. Let $R$ denote the set of all rational flips, i.e., $R=\{(g,g^*): g \in A_i, g^* \in A_j, v_i(g^*) > v_i(g)\}$.

We will construct three  sets $A_i^*$, $A_j^*$ and $R^*$ iteratively, using the following procedure: At first all sets are empty, i.e., $A_i^*=A_j^*=R^*=\emptyset$. Then, iteratively we select $g^* \in \arg\max_{g \in A_j\setminus A_j^*} v_i(g)$ and we find an item $g \in A_i\setminus A_i^*$ such that $v_i(g^*)>v_i(g)$. If there are multiple such items we select the maximum w.r.t to $v_i(\cdot)$. Then we allocate: $g \in A_i^*$, $g^* \in A_j^*$ and $(g,g^*) \in R^*$.
The procedure ends when for some $g^* \in \arg\max_{g \in A_j\setminus A_j^*} v_i(g)$ we cannot find any $g \in A_i\setminus A_i^*$ such that $v_i(g^*) > v_i(g)$.

Observe that the above process implies that $|A_i^*|=|A_j^*|$. Also, note that for any items $g \in A_i\setminus A_i^*$ and $g^* \in A_j\setminus A_j^*$ it must be that $v_i(g) \geq v_i(g^*)$, due to the stopping conditions in the above process.  
Hence,

\begin{align}
    v_i(A_i\setminus A^*_i) - v_i(A_j\setminus A_j^*) = \sum_{g \in A_i\setminus A_i^*} v_i(g) - \sum_{g^* \in A_j\setminus A_j^*} v_j(g^*) \geq 0.\nonumber
\end{align}

This implies that
\begin{align}\label{eq:EFF1_NASH_1}
    v_i(A_j) - v_i(A_i) = v_i(A^*_j) - v_i(A^*_i) + v_i(A_j\setminus A_j^*) - v_i(A_i \setminus A_i^*) \leq  v_i(A^*_j) - v_i(A^*_i).
\end{align}

Hence,

\begin{align}\label{eq:EFF1-NASH2}
    \sum_{(g,g^*) \in R^*} v_i(g^*) -v_i(g) &= v_i(A^*_j) - v_i(A^*_i) \nonumber \\
    & \geq v_i(A_j) - v_i(A_i) \nonumber \\
    & \geq v_i(A_j) + 3\cdot\max_{(g,g^*) \in R}\{ v_i(g^*) - v_i(g) \}.
\end{align}

The top inequality is due to~(\ref{eq:EFF1_NASH_1}). The last inequality is due to (\ref{eq:EFF1_NASH_1_contradiction}), i.e., the assumption that the allocation is not $1/2$-\efOne.
Observe also that:

\begin{align}\label{eq:EFF1-NASH3}
    \sum_{g,g^* \in R^*} v_j(g^*) - v_i(g) = v_i(A^*_j) - v_i(A^*_i) \leq v_i(A^*_j) \leq v_i(A_j).
\end{align}

Let $(g,g*) = \arg\min_{(g,g^*) \in R} \frac{v_j(g^*)-v_j(g)}{v_i(g^*)-v_j(g)}$. Then 

\begin{align}\label{eq:EFF1-NASH4}
    \frac{v_j(A_j)}{v_i(A_i) + v_i(g^*) - v_i(g)} 
    &\geq \frac{v_j(A_j)}{v_i(A_i) + 3\cdot v_i(g^*) - 3\cdot v_i(g)} \nonumber \\
    &> \frac{\sum_{(g',g'') \in R^*} v_i(g'') - v_i(g')}{\sum_{(g',g'') \in R^*} v_i(g'') - v_i(g')} \nonumber \\
    & \geq \frac{v_j(g^*)-v_j(g)}{v_i(g^*)-v_i(g)}.
\end{align}

The first inequality is due to $\frac{x}{y+a} > \frac{x}{y+b}$ for any $x\geq0, y\geq 0, b>a$ and $a\geq 0$, and that $v_i(g^*)-v_i(g)>0$. The second inequality is due to (\ref{eq:EFF1-NASH3}) for the nominator and (\ref{eq:EFF1-NASH2}) for the denominator.
Observe now that we can re-write (\ref{eq:EFF1-NASH4}) as:

\begin{align}
    \frac{v_j(A_j)}{v_i(A_i) + v_i(g^*) - v_i(g)} > \frac{v_j(A_j) -v_j(A_j) + v_j(g^*)-v_j(g)}{v_i(A_i) -v_i(A_i) + v_i(g^*)-v_i(g)} = \frac{v_j(A_j) - v_j(A_j \cup g \setminus g^*)}{v_i(A_i \cup g^* \setminus g) -v_i(A_i)}.
\end{align}

By rearrangement, we get

\begin{align}
        v_j(A_j)\left(v_i(A_i \cup g^* \setminus g) -v_i(A_i)\right) &> \left(v_j(A_j) - v_j(A_j \cup g \setminus g^*)\right)v_i(A_i \cup g^* \setminus g) \nonumber \\
        & \iff \nonumber \\
        v_j(A_j)\cdot v_i(A_i)&<v_i(A_i \cup g^* \setminus g)\cdot v_j(A_j \cup g \setminus g^*). \label{eq:MNWcontradicted}
    \end{align}

Consider now the alternative allocation $\aaa'_{-\{i,j\}}=\aaa$ and $A'_i=(A_i \cup g^* \setminus g)$, $A'_j=(A_j \cup g \setminus g^*)$; if $\prod_{i=1}^n v_i(A_i) > 0$ then clearly $\aaa'$ has a strictly higher Nash welfare, contradicting the maximality of $\aaa$. Therefore, we focus on the case $\prod_{i=1}^n v_i(A_i) = 0$.  If any of $v_i(A_i)=0$ or $v_j(A_j)=0$ holds, then inequality~(\ref{eq:MNWcontradicted}) implies that both $v_i(A_i \cup g^* \setminus g) > 0$ and $v_j(A_j \cup g \setminus g^*)>0$, and as such, $\aaa'$ has strictly more non-zero valued agents, again contradicting the maximality of $A$. Finally if $\prod_{i=1}^n v_i(A_i) = 0$ and both $v_j(A_j)>0$, $v_i(A_i)>0$ then $\aaa'$ yields a strictly higher Nash welfare for the non-zero valued agents, contradicting the maximality of $\aaa$.
\end{proof}

It is well known that the maximum Nash welfare always returns PO allocations (see e.g., \cite{caragiannis2019unreasonable}). The following theorem is given for the sake of completeness and proves this for the case of $k$-sized bundles.

\begin{theorem}\label{thm:PO:MNWisPO}
Any $k$-bundle MNW-optimal allocation is PO.
\end{theorem}

\begin{proof}
    Suppose, for the sake of contradiction, there exists an allocation \(\mathcal{A}'\) such that \(v_i(A_i') \geq v_i(A_i)\) for all \(i \in N\), and \(v_j(A_j') > v_j(A_j)\) for some $j \in N$. Let $P(\mathcal{A}) = \{ i \in N : v_i(A_i) > 0 \}$ and $P(\mathcal{A}') = \{ i \in N : v_i(A_i') > 0 \}$. If \(|P(\mathcal{A}')| > |P(\mathcal{A})|\), then \(\mathcal{A}'\) contradicts the maximality of \(P(\mathcal{A})\).  
If \(|P(\mathcal{A}')| = |P(\mathcal{A})|\), then the strict improvement implies
\[
\prod_{i \in P(\mathcal{A}')} v_i(A_i') > \prod_{i \in P(\mathcal{A})} v_i(A_i),
\]
contradicting the optimality of \(\mathcal{A}\) under the MNW rule, and the theorem follows.
\end{proof}

In the following theorem, we present an example showing that this bound is asymptotically tight, even for some notable special cases.

\begin{theorem}
    There exist instances where a $k$-bundle MNW-optimal allocation is no better than $1/2$-\efOne.
\end{theorem}

\begin{proof}

Consider the instance with two agents and $2k$ items, namely $g_1,...,g_{2k}$, partitioned into two sets $G_1=\{g_1,..,g_k\}$ and $G_2=\{g_{k+1},..,g_{2k}\}$, with the following valuation functions:
For any $c > 0$, $v_1(g_\ell)=c$ for $\ell \in [k]$ and $v_1(g_\ell)=0$ otherwise. Similarly, $v_2(g_\ell)=c$ for $g \in [k]$ and $v_2(g_\ell)=1$ otherwise (shown also in Table~\ref{tab:MNWlowerbound}).

\begin{table}[tb]
    \centering
    \begin{tabular}{c|cccccccc}
    & $g_1$ & $g_2$ & $\cdots$ & $g_k$ & $g_{k+1}$ & $\cdots$ & $g_{2k}$ \\ \hline
$v_1(\cdot)$ & $c$ & $c$ & $\cdots$ & $c$ & $0$ & $\cdots$ & $0$ \\
$v_2(\cdot)$ & $2$ & $2$ & $\cdots$ & $2$ & $1$ & $\cdots$ & $1$ \\
\end{tabular}
    \caption{An upper bound instance on \efOne for a maximum Nash welfare solution.}
    \label{tab:MNWlowerbound}
\end{table}

We will show that the maximum Nash welfare is achieved only by the allocation $A_1= (g_1,...,g_k)$ and $A_2 = (g_{k+1},...,g_{2k})$. Let $z$ be the number of items in $G_1=\{g_1,...,g_k\}$ allocated to agent $1$. Then agent $1$ receives $k-z$ items  from the second group $G_2=\{g_{k+1},...,g_{2k}\}$. Consequently, agent 2 receives the remaining $k-z$ items from $G_1$ and $z$ from  $G_2$. The resulting Nash welfare is therefore 

\begin{align}
    (cz)\cdot(2(k-z)+z). \label{eq:NMconstruction:1}
\end{align}

The first derivative of~(\ref{eq:NMconstruction:1}) w.r.t to $z$ is $2c(k-z)$ which is equal to $0$ only when $z=k$, hence the allocation $(A_1,A_2)$ is the unique allocation maximizing the Nash welfare. Let $a \in A_1$ and $b \in A_2$, then: 

\begin{align}
    v_2(A_2 \setminus b \cup a) = (k-1) + 2 = k+1 < \left(\frac{k+1}{2k-1} + \epsilon\right) 2k-1 =  2(k-1) + 1 =  v_1(A_1 \setminus a \cup b)\nonumber
\end{align}

for any $\epsilon>0$, hence this allocation is at most $\frac{k+1}{2k-1}$-\efOne, which approaches $1/2$-\efOne as the bundles size $k$ increases.
\end{proof}

Note that the example precludes the possibility of a better approximation for the Nash welfare algorithm for many reasonable assumptions. First, the construction works even for $n=2$ and naturally extends to any number of agents. Second, it holds even in the case where the agents \emph{agree on the rankings} of the items -- this comes at a stark contrast with the case where all the agents \emph{agree and on value} of the item where the same mechanism returns 1-\efx and PO allocations~\cite{Bog2024}. Moreover, by setting $c=3$ the impossibility result holds even when the agents are normalized\footnote{meaning that $\sum_{g \in M} v_i(g)=C$ for some constant $C$ and for all $i \in N$.} and by setting $c=2$ it continues to hold when each agent's valuation is restricted to at most three distinct values. Again, this comes in contrast with the binary case where 1-\efx is achievable by the maximum Nash welfare solution. 

Another rule that always returns PO allocation is the leximin criterion. The \emph{leximin} rule seeks to maximize fairness by prioritizing the worst-off agent, then recursively applying the same principle to the remaining agents: it first maximizes the value of the least well-off agent, then of the second least well-off among the rest, and so on. Leximin looks somewhat promising since when $n=2$ or when all agents are identical, it guarantees \efx allocations, although it cannot guarantee \efOne allocations for general additive valuations. Here we show that unfortunately, this rule fails to guarantee \efx allocations even under the common ranking assumption, and cannot yield any $\gamma$-\efOne allocation for any $\gamma \in O(1/k)$.

Finally, we explore the  \emph{social welfare} (SW) maximizing rule, which aims to maximize the total value, selecting an allocation $\aaa \in \arg\max \sum_{i=1}^n v_i(A_i)$, and is arguably the most common efficiency criterion. Probably less striking, this rule also fails to guarantee any approximation better than $O(1/k)$.

In the following, we expand an example from~\cite{Bog2024} to show that both the leximin and social welfare maximizing rules cannot yield any approximation significantly better than the trivial bound of $1/k$. Interestingly, the same example returns an $1/2$-\efOne approximation for the MNW rule, but we have omitted it since we were able to establish that with a simpler example. The proof is deferred to Appendix~\ref{app:PO}.

\begin{theorem}\label{thm:PO:SWleximin-inapprox}
There exists a leximin-optimal allocation that is $\left(\frac{2}{k} + o\left(\frac{1}{k^2}\right)\right)$-\efOne. 
Moreover, there exists a social welfare-optimal allocation that is $\frac{1}{k - 1}$-\efOne.
\end{theorem}

We have also examined the existence of \efx allocations under the common ranking assumption. A prior example by~\cite{Bog2024} shows that exact \efOne allocations are not possible even under the common ranking assumption for both SW and MNW. For the leximin rule, we provide the following example showing that leximin cannot return \efx allocations, even in the common ranking regime. The proof is deferred to Appendix~\ref{app:PO}.

\begin{theorem}\label{thm:PO:LeximinFailsEfxSameRanking}
    There is an instance where no leximin allocation is also \efx, even for $n=3$ and under the common ranking assumption.
\end{theorem}

\section{Conclusions}
We have studied a constrained model of fair division, where all agents must receive the same number of items. Motivated by the work of \cite{Bog2024}, which introduced notions of envy-freeness up to rational flips, we studied further these criteria from the angle of approximation algorithms. Our results reveal that, as in the unconstrained model, the concept of \efx is again quite challenging and it remains an open question to come up with bounded approximation guarantees for general additive valuations. At the same time, \efOne is a more tractable criterion, which is also more compatible with some notions of efficiency, such as the Nash welfare, as discussed in Section~\ref{sec:PO}.

We conclude with a  discussion on extending this framework beyond additive valuations. We note that the particular notion of a rational flip, that we use here and in \cite{Bog2024} for defining  \efOne and \efx, seems to be suitable only for additive valuations. In particular, additivity ensures that if an agent $i$ envies another agent $j$, then there is always at least one rational flip available that can improve the value for agent $i$. But with non-additive valuations, it is possible that agent $i$ envies agent $j$, yet there is no flip that could improve $i$.  Notably, simply exchanging a single pair of items between bundles can have unexpected effects on an agent’s valuation of a bundle. This is in contrast to the standard setting, and the notions of EF1 and EFX, where for any monotone valuation, the thought experiment of removing an item from the bundle of agent $j$ results in a better situation for agent $i$. Therefore, under our constrained setting, one would need to adapt or extend  the notion of a rational flip for non-additive functions to ensure the existence of {\it corrective} actions in the underlying thought experiment. Developing this further  remains an interesting direction for future work.

\section*{Acknowledgements}
This work has been partially supported by project MIS 5154714 of the National Recovery and Resilience Plan Greece 2.0 funded by the European Union under the NextGenerationEU Program.
It has also been supported by the H.F.R.I call ``Basic research Financing (Horizontal support of all Sciences)'' under the National Recovery and Resilience Plan Greece 2.0 funded by the European Union - NextGenerationEU (H.F.R.I. Project Number: 15877). This research was partially undertaken during the first author’s internship in Archimedes Unit at Athena Research Center.

\bibliographystyle{ACM-Reference-Format}
\bibliography{mybibliography}
\newtheorem*{customtheorem}{Theorem}
\newtheorem*{customlemma}{Lemma}

\appendix

\section{Envy cycle elimination in the unconstrained setting}
\label{app:ece}

For the sake of completeness, in Algorithm~\ref{alg:ece-standard} we present the classic algorithm by \cite{lipton2004approximately} for allocating items by finding an unenvied agent in each round via envy cycle elimination.

\begin{algorithm}[ht]
\caption{Envy cycle elimination algorithm}\label{alg:ece-standard}
\begin{algorithmic}[1] 
\Procedure{EnvyCycleElimination}{$N,M$}

\vspace{0.5em}

\For{$i \in N$}
    \State {$A_i \gets \emptyset$}
\EndFor

\vspace{0.5em}

\While{there exist unallocated items}
    \State Let $G$ be the envy graph, i.e., $G = (N, \{(i, j) \mid v_i(A_i) < v_i(A_j)\})$
    \State $j \gets$ \Call{FindUnenviedAgent}{$G$} \Comment {Guaranteed to exist; ties break arbitrarily}
    \State {Agent $j$ adds to $A_j$ her favorite unallocated item}
    \If {the envy graph has no source vertex}
        \State{Remove envy cycles until a source vertex is created}
    \EndIf
\EndWhile

\vspace{0.5em}
\State \Return {$(A_1, A_2,\ldots,A_n)$}
\EndProcedure
\end{algorithmic}
\end{algorithm}

\section{Missing proofs from Section~\ref{sec:RR}}
\label{app:RR}

\subsection{Proof of Theorem~\ref{thm:RR:Eff1}}

\begin{customtheorem}[\ref{thm:RR:Eff1}]
    The allocation computed by the Round-Robin algorithm satisfies \efOne.
\end{customtheorem}

\begin{proof}
    Fix a pair of agents $i,j$. We will analyze the envy of agent $i$ towards $j$.
    Consider the bundles $A_i = (a_1, a_2,\ldots, a_k)$ and $A_j = (b_1, b_2,\ldots, b_k)$ allocated to agents $i$ and $j$, respectively, by the Round-Robin algorithm, where for each $\ell\in [k]$, the items $a_\ell$ and $b_\ell$ were allocated in the $\ell$-th round of the algorithm. 
    
    If $i$ precedes $j$ in the Round-Robin order, then $i$ always picks her favorite remaining item before $j$ and therefore there is no envy towards $j$. 
    Hence, the only interesting case is when $j$ precedes $i$ in the Round-Robin order. 
    In this case, we know that when agent $i$ picks an item at round $\ell$, it is at least as valuable as what $j$ picked at around $\ell+1$. Therefore, we have $v_i(a_\ell) \geq v_i(b_{\ell+1})$ for all $\ell \in \{1,\ldots,k-1\}$. Hence, 
    \begin{align}\label{ef1: ineq1}
        v_i(A_i \mminus a_k) \geq v_i(A_j \mminus b_1)
    \end{align}
    
    If $v_i(A_i) \geq v_i(A_j)$ then once again agent $i$ does not envy agent $j$ and \efOne is not violated.
    Otherwise, if $v_i(A_i) < v_i(A_j)$, then due to (\ref{ef1: ineq1}) it must hold that 
    \begin{align}\label{ef1: ineq2}
        v_i(b_1) > v_i(a_k)
    \end{align}
    
     Combining (\ref{ef1: ineq1}) with (\ref{ef1: ineq2}) yields
    \begin{align*}
        v_i(A_i \mminus a_k \pplus b_1) \geq v_i(A_j \mminus b_1 \pplus a_k)
    \end{align*}
    which implies that the \efOne condition is satisfied, regarding the envy of $i$ towards $j$. This concludes the proof.    
\end{proof}

\section{Missing proofs from Section~\ref{sec:ece}}
\label{app:modified_ece}

\subsection{Proof of Lemma~\ref{lemma: common_top_n_min_value}}

\begin{customlemma}[\ref{lemma: common_top_n_min_value}]
    Let $i \in [n]$, and let $B_i \neq \emptyset$ be a bundle held by agent~$i$ at some point after the first $n$ items have been allocated. Then it holds that $v_i(B_i) \geq v_i(g_i^n)$. Moreover, when all agents agree on the set of top-$n$ items, no two such items ever appear in the same bundle.
\end{customlemma}

\begin{proof}
    During the first $n$ iterations of the algorithm, there are no privileged agents ($P=\emptyset$) and the tie-breaking rule prioritizes agents with fewer items. As a result, in each of these rounds, an agent with an empty bundle is selected and allocated her most preferred item from the remaining pool $S$. Since there are $n$ agents and $n$ top-ranked items, each agent receives at worst her $n$-th most preferred item. Therefore, by Lemma~\ref{lemma: agents-improve}, it follows that after the first $n$ items have been allocated, any bundle $B_i$ held by agent $i$ satisfies $v_i(B_i) \geq v_i(g_i^n)$, for all $i \in [n]$.

    For the second part of the lemma, we show that no two top-$n$ items are ever placed in the same bundle. This is clearly true at the end of the first $n$ rounds, as each top-$n$ item is allocated to a distinct agent. We now argue that this invariant is maintained throughout the algorithm. To see this, note that the only way that such an item could end up in a bundle with another item from $T^n$ is that some agent who holds already an item from $T^n$ becomes privileged and drops it by performing an item swap. Consider the first time that this happens. Since we know that in an item swap, an agent \emph{exchanges} her least valued item with her most preferred item from $S$, and since at that point, there is at most one item from $T^n$ in her bundle, it follows that it will not be removed in an item swap\footnote{In the extreme case where for an agent $j$, the items not in $T^n$ have the same value as the item that $j$ holds from $T^n$, we assume that she will use a tie-breaking rule that will choose to swap an item not from $T^n$.}. Therefore, inductively, all agents will maintain throughout the algorithm an item from $T^n$ in their bundle.

    Finally, observe that during possible envy-cycle eliminations, bundles may change ownership, but their internal contents remain intact—no items are moved between bundles.
    Hence, throughout the algorithm, no two top-$n$ items are ever placed in the same bundle.
\end{proof}

\subsection{Proof of Lemma~\ref{lem:polynomial}}

\begin{customlemma}[\ref{lem:polynomial}]
    The algorithm always terminates after a polynomial number of iterations.
\end{customlemma}

\begin{proof}
We begin by observing that as long as $S \neq \emptyset$, there exists at least one agent whose bundle has fewer than $k$ items. This is because the total number of items is finite and equal to $|M| = nk$, and thus if all bundles were full, $S$ would be empty. Moreover, according to the algorithm, an agent may \emph{get} an item from the pool $S$ only when her current bundle is not full; otherwise, she may either \emph{swap} or \emph{pass}.

Now, consider a round in which the algorithm selects an agent $j$ with bundle $B_j$. Depending on the state of $j$ and the contents of her bundle, one of the following cases applies:

\begin{description}
    \item[Case 1:] $j \in N \setminus P$ and $|B_j| < k$. In this case, agent $j$ performs a \emph{get} operation (receives an item from $S$), reducing the size of the pool. Since $S$ is finite, this can happen only a finite number of times.

    \item[Case 2:] $j \in N \setminus P$ and $|B_j| = k$, and $j$ performs a \emph{swap}. This operation replaces her least valued item with a more preferred one from $S$, strictly increasing her value. As the set of possible bundles is finite and preferences are fixed, an agent cannot perform infinitely many strictly improving swaps—eventually, she will hold her top-$k$ items and no longer have an incentive to swap.

    \item[Case 3:] $j \in P$ and $|B_j| = k$, and $j$ performs a \emph{swap}. The algorithm allows agents in $P$ to become active again if a beneficial swap is possible. When this happens, $j$ strictly increases her value and leaves $P$. However, each such swap corresponds to an improvement, and as previously argued, no agent can improve indefinitely.

    \item[Case 4:] $j \in N \setminus P$ and $|B_j| = k$, and $j$ performs a \emph{pass} operation, becoming privileged and joining the set $P$. In this case, no change occurs in $S$, but the number of active agents (those in $N \setminus P$) decreases. This process can continue only until all agents with full bundles have either performed a swap or passed and joined $P$. Eventually, the only remaining agents in $N \setminus P$ will be those whose bundles are not yet full. At that point, the algorithm must select one of these agents, who will then perform a \emph{get} operation, removing an item from $S$ and thereby progressing the algorithm toward termination.

\end{description}

Since all operations either reduce the number of items in the pool $S$ (gets), reduce the number of possible improving actions (swaps) or reduce the number of active agents (passes), the algorithm terminates after a finite number of steps.

Furthermore, the algorithm runs in polynomial time with respect to the parameters $k$ and $n$. This follows from the observation that all operations fall into one of three categories: operations that strictly increase an agent’s value (\textit{get} and \textit{swap}), and those that do not yield any value improvement (\textit{pass}). Between two successive value-increasing operations, there can be at most $n$ consecutive \textit{pass} operations, as each agent may pass at most once before a value-improving step occurs. Additionally, each time an agent holds a specific bundle, the number of \textit{swap} operations she can perform is bounded by $kn$, since each agent can only consider exchanging each item at most once, and there are $kn$ items in total. Finally, the number of \textit{get} operations is exactly $kn$, as each item is allocated once.

\end{proof}

\section{Missing proofs from Section~\ref{sec:PO}}
\label{app:PO}

\subsection{Proof of Theorem~\ref{thm:PO:SWleximin-inapprox}}

\begin{customtheorem}[\ref{thm:PO:SWleximin-inapprox}]
There exists a leximin-optimal allocation that is $\left(\frac{2}{k} + o\left(\frac{1}{k^2}\right)\right)$-\efOne. 
Moreover, there exists a social welfare-optimal allocation that is $\frac{1}{k - 1}$-\efOne.
\end{customtheorem}

\begin{proof}
We will prove this statement by expanding an example from~\cite{Bog2024}. Consider an instance with $n$ agents, and $kn$ items. The items are of $4$ types: $k$ type $a$ items, $k$ type $b$, $n-2$ type $c$ items, and the remaining are items of type $d$. 

The instance includes two special agents, $i$ and $j$, and a set $L$ of $n-2$ agents with identical valuations. The valuations for agent $i$ are, for some $x \in [0,k]$, $v_i(a)=k-x$, $v_i(b)=x$, $v_1(c)=v_1(d)=0$. For agent $j$ the valuations are $v_j(a)=1$, $v_j(b)=0$, $v_j(c)=\frac{k^2}{n-2}$ and $v_j(d)=0$. Finally, for any agent $\ell \in L$,  the valuations are $v_\ell(c)=\frac{k(k+1)}{n-2}$ and $0$ in all other cases. Hence all agents have total value equal to $k(k+1)$. 

In the following we will show that the leximin criterion must allocate the items as follows: $A_i=\{b,...,b\}$\footnote{For the sake of simplicity, in this example we consider the allocations as multi-sets.}, $A_j=\{a,...,a\}$ and $A_\ell=\{c,d,...,d\}$. This allocation fails to be $\gamma$-\efOne when

\begin{align}
    v_i(A_i \mminus b \pplus a) = (k-1)x + k - x & < \gamma \cdot ((k-x)(k-1) + x)\nonumber = \gamma \cdot v_i(A_j \pplus a \mminus b), \nonumber
\end{align}

i.e., when 
\begin{align}
    \gamma \geq \frac{k-2x+kx}{(k-x)(k-1)+x}. \label{ineq:gamma_NW}
\end{align} Observe that this bound is increasing with $x$. For the leximin criterion, we will show that this allocation holds for every $x>1$, and as such $\gamma < \frac{2k-1}{k^2-2k+2} \leq \frac{2}{k} + o(1/k^2)$.

First, notice that any leximin allocation should give one item $c$ to each agent $\ell \in L$; otherwise some agent would receive zero value, which violates leximin criterion, since there exists at least one allocation (i.e., allocation $\aaa$) that avoids this situation. In this way, we construct $n-2$ bundles $A_\ell=\{c,d,...,d\}$. We are left with $k$ items of type $a$ and $k$ of type $b$, to be allocated between agents $i$ and $j$. Suppose $z > 0$ items of type $a$ are given to agent $i$, and the remaining $k - z$ to agent $j$. Then, agent $i$'s value is $z(k - x) + (k - z)x$, and agent $j$'s value is $k - z$. For any $z > 0$, we have $v_j < k$, and transferring one $a$-item from $i$ to $j$ strictly increases the smaller of the two values, contradicting leximin optimality. Therefore, agent $j$ must receive all $k$ items of type $a$.

For the social welfare maximization rule, we show a different allocation, when $x=0$. 
The bulk of the analysis is that even if we give the maximum possible value to $i$ and $j$, this is at most $k$, while simply giving one type $c$ item to the $n-2$ agents $\ell \in L$ leads to the total value of at least $k(k+1)$. Hence, type $c$ items should be given to type $\ell$ agents. Let $z$ be the number of type $a$ items allocated to agent $i$. Hence $k-z$ is the number of type $a$ items allocated to $j$ and the number of type $b$ items allocated to $i$, and the total welfare, given this parameter is $zk + 2(k-z)$ which is maximized for $z=k$. Therefore, agent $i$ gets the bundle $A'_i=\{a,...,a\}$ and agent $j$ gets the bundle $A'_j=\{b,...,b\}$. 
Thus, 

\begin{align}
    v_j(A'_j \pplus a \mminus b ) = 1 < \gamma\cdot (k-1) =  v_j(A'_i  \pplus b \mminus a).
\end{align}

Thus, the social welfare maximizing rule is $\frac{1}{k-1}$-\efOne in this instance.
\end{proof}

\subsection{Proof of Theorem~\ref{thm:PO:LeximinFailsEfxSameRanking}}

\begin{customtheorem}[\ref{thm:PO:LeximinFailsEfxSameRanking}]
    There is an instance where no leximin allocation is also \efx, even for $n=3$ and under the common ranking assumption.
\end{customtheorem}

\begin{proof}
    We construct an instance with 3 agents and 9 items. The table below shows the valuations of the 9 items for agents $1$ and $2$ and $3$.

    \begin{center}
    \begin{tabular}{l|c c c c c c c c c}
        & $g_1$ & $g_2$ &$g_3$ & $g_4$ & $g_5$ & $g_6$ & $g_7$& $g_8$& $g_9$ \\
        \hline
         $v_1(\cdot)$&\cellshade 50& 17 &16 &14 &2 &1 & 0& \cellshade 0& \cellshade 0  \\
         $v_2(\cdot)$&46& 17 &\cellshade 16 &\cellshade 15 &3 &\cellshade 3 & 0& 0& 0  \\
         $v_3(\cdot)$&33 & \cellshade 17 &15 &15 &\cellshade 11 &4 & \cellshade 3& 1& 1  \\
    \end{tabular}        
    \end{center}

    First, note that the allocation $\aaa = (A_1, A_2, A_3)$ (shaded in the table above), where $A_1 = (g_1, g_8, g_9)$, $A_2 = (g_3, g_4, g_6)$, and $A_3 = (g_2, g_5, g_7)$, yields a value vector of $(v_1(A_1),v_2(A_2,v_3(A_3))=(50, 34, 31)$. However, this allocation is not \efx, as $v_3(A_3) = 17 + 11 + 3 = 31 < 15 + 15 + 4 = 34 = v_3(A_2)$, meaning agent 3 envies agent 2. Even after performing a rational flip of items $(g_7, g_6)$, we have
    \begin{align*}
        v_3(A_3 + g_6 - g_7) = 17 + 11 + 4 = 32 < 33 = v_3(A_2 - g_6 + g_7).
    \end{align*}
    
    \noindent Hence, there exists a rational flip that does not eliminate the envy.

    Next, we will demonstrate that the allocation above is the unique leximin allocation. More precisely, we will show that in any allocation $\aaa'$ where $v_3(A'_3) > 31$, the value vector is leximin-dominated by $(50, 34, 31)$.
    This is achieved through a case analysis on $A_3$:
    \begin{description}
        \item[Case 1:] If $g_1 \in A'_3$, assume $v_1(A'_1) \geq 31$. This implies that at least two of $\{g_2, g_3, g_4\}$ are allocated to $A'_1$, leading to $v_2(A'_2) \leq 17 + 3 + 3 < 31$. Now, if $v_2(A'_2) \geq 31$, then at least two of $\{g_2, g_3, g_4\}$ are allocated to $A'_2$, resulting in $v_1(A'_1) \leq 17 + 2 + 1 < 31$.  
        \item[Case 2:] If $g_2, g_3 \in A'_3$, then at least one of the agents 1 or 2 receives a value of at most $14 + 2 + 1 < 31$ or $15 + 3 + 3 < 31$.
        \item[Case 3:] If $g_2, g_4 \in A'_3$, then similarly, either $v_1(A_1') < 16 + 2 + 1$ or $v_2(A_2') < 16 + 3 + 3$.
        \item[Case 4:] If $g_3, g_4 \in A'_3$, then either $v_1(A_1') < 17 + 2 + 1$ or $v_2(A_2') < 17 + 3 + 3$.
        \item[Case 5:] If $g_2, g_5, g_6 \in A'_3$, and $g_1 \in A'_1$, then $v_2(A_2') \leq 16 + 15 = 31$. The vector $(50, 32, 31)$ is leximin-dominated by $(50, 34, 31)$. If $g_2$ is in $A'_2$, then $v_1(A_1') \leq 16 + 14 < 31$.
    \end{description}

    In all remaining allocations, agent 3's value remains less than 31, meaning no other allocation leximin-dominates $(50, 34, 31)$.

    Now, observe that an \efx allocation exists. Consider the allocation $\bbb = (B_1, B_2, B_3)$, where $B_1 = (g_1, g_8, g_9)$, $B_2 = (g_3, g_4, g_7)$, and $B_3 = (g_2, g_5, g_6)$. In this allocation, agent 1 does not envy anyone. Agent 2 envies agent 1. However, any rational flip involves $g_1$ and, hence, eliminates the envy. Agent 3 envies both agent 1 and agent 2. Regarding agent's $3$ envy towards agent $1$, any rational flip involves item $g_1$ and also eliminates the envy. As for the envy towards agent 2, we need to consider the rational flip $(g_5, g_4)$, which minimizes the value gain for agent $3$. (If this flip satisfies \efx, then no other flip will violate it). Indeed, performing this rational flip removes the envy, proving that the allocation $\bbb$ is \efx. 
\end{proof}

\end{document}